\documentclass[11pt]{article}
\usepackage[margin=2.3cm, includehead]{geometry}

\usepackage{currfile}

\usepackage{natbib}
\usepackage{makeidx}
\usepackage{multirow}
\usepackage{multicol}
\usepackage[dvipsnames,svgnames,table]{xcolor}
\usepackage{graphicx}
\usepackage{epstopdf}
\usepackage{ulem}
\usepackage[bookmarks]{hyperref} 
\usepackage{amsmath}
\usepackage{amsfonts}
\usepackage{amssymb}
\usepackage{graphics}
\usepackage{epsfig}
\usepackage{verbatim}
\usepackage{amsthm}
\usepackage{geometry}
\usepackage{framed} 
\usepackage{scalefnt}
\usepackage{alltt}
\usepackage{hyperref}
\hypersetup{
 bookmarks=true, 
 unicode=false, 
 pdftoolbar=true, 
 pdfmenubar=true, 
 pdffitwindow=true, 
 pdftitle={My title}, 
 pdfauthor={Author}, 
 pdfsubject={Subject}, 
 pdfnewwindow=true, 
 pdfkeywords={keywords}, 
 colorlinks=true, 
 linkcolor=blue, 
 citecolor=blue, 
 filecolor=magenta, 
 urlcolor=cyan 
}

\def\iid{\buildrel {\rm i.i.d.} \over \sim}

\def\i.i.d.{\buildrel {\rm i.i.d.} \over \sim}

\def\cw#1 { \overset{\mathbb{P}}{\underset{#1}{\longrightarrow}} }

\def\Natu0{\mathbb{N}_0}
\def\P#1{{\mathrm{P}}\left(#1\right)}

\def\E#1{{\mathrm E}\left[#1\right]}

\def\Var#1{{\mathrm Var}\left(#1\right)}

\def \rcov#1#2 {{\rm cov}_{#1}\left( #2\right)}

\DeclareMathOperator*{\argmin}{arg\,min}

\newtheorem{lemma}{Lemma}
\newtheorem{theorem}{Theorem}

\newtheorem{corollary}{Corollary}

\begin{document}
\title{ Designing a Bonus-Malus system reflecting the claim size under the dependent frequency--severity model \\
 }

%
%
%
%
%
\author{Rosy Oh\footnote{First authors (rosy.oh5@gmail.com),  Institute of Mathematical Sciences, Ewha Womans University, Seoul, Republic of Korea.}, Joseph H.T. Kim\footnote{Corresponding authors (jhtkim@yonsei.ac.kr), Department of Applied Statistics, College of Business and Economics, Yonsei University, Seoul, Republic of Korea.}, Jae Youn Ahn\footnote{Corresponding authors (jaeyahn@ewha.ac.kr), Department of Statistics, Ewha Womans University, Seoul, Republic of Korea.}}


\maketitle

\begin{abstract}

In auto insurance, a Bonus-Malus System (BMS) is commonly used as a posteriori risk classification mechanism to set the premium for the next contract period based on a policyholder's claim history. Even though recent literature reports evidence of a significant dependence between frequency and severity, the current BMS practice is to use a frequency-based transition rule while ignoring severity information. Although \citet{PengAhn} claim that the frequency-driven BMS transition rule can accommodate the dependence between frequency and severity, their proposal is only a partial solution, as the transition rule still completely ignores the claim severity and is unable to penalize large claims. In this study, we propose to use the BMS with a transition rule based on both frequency and size of claim, based on the bivariate random effect model, which conveniently allows dependence between frequency and severity. We analytically derive the optimal relativities under the proposed BMS framework and show that the proposed BMS outperforms the existing frequency-driven BMS. Later numerical experiments are also provided using both hypothetical and actual datasets in order to assess the effect of various dependencies on the BMS risk classification and confirm our theoretical findings.

\end{abstract}

Dependence, Generalized linear model, Rate making, Bonus-Malus system, bivariate random effects model 



\vfill

\pagebreak

\vfill

\pagebreak

\section{Introduction}
%

 In the auto insurance industry, the ratemaking process begins by classifying a given policyholder into different risk classes at the outset of the contract. This initial or a priori ratemaking process is then followed by a posteriori ratemaking system in subsequent policy years based on the claim history of the policyholder. For practitioners, this a posteriori ratemaking system for the prediction of a premium based on a policyholder's claim history and it is known as the Bonus-Malus System (BMS). The BMS system typically constitutes three elements: the Bonus-Malus (BM) levels, transition rules to navigate the BM levels, and the relativity attached to each BM level. In the classical BMS, the transition rule is governed by frequency; it ignores the severity information, which implies that the future loss of a policyholder can be appropriately modeled by predicting the frequency only (see, e.g., \citet{Lemaire2, Denuit2, Chong}. In reality, this frequency-driven BMS transition rule is the standard practice in many jurisdictions. It carries an implicit assumption that the frequency and severity are independent, such that the premium can be simply computed as a product of the mean frequency and mean severity. Theoretically, the structure of the traditional BMS is directly related to the classical theory of the collective risk model, which assumes independence between frequency and severity for mathematical tractability and convenience \citep{Klugman}.\\

However, a series of recent empirical studies have shown that the dependence between frequency and severity in auto insurance is statistically significant (\citet{Frees4,Garrido}). This phenomenon invalidates the practice of using frequency-driven BMS and highlights the need to extend the classical collective risk model by allowing some dependence structure between frequency and severity. Existing studies on dependent frequency--severity models and the associated insurance premiums include copula-based models \citep{Czado, Frees4}, two-step frequency--severity models \citep{Frees2, Peng, Garrido, Park2018does}, and bivariate random effect-based models \citep{Pinquet, Cossette, Bastida, Czado2015, Lu2016, JKWoo}.

The random effect model is especially popular in insurance ratemaking because of the mathematical tractability in its prediction. The bivariate random effect model consists of two random effect components. The first random effect induces the dependence among frequencies and the second induces the dependence among individual severities. These two random effects are then jointly modeled to induce the dependence between frequency and severities at the distribution level. Statistical methods for prediction using the random effect model are well developed in the statistical and insurance literature, such as in \citet{Pinquet, Cossette, Bastida, Czado2015, Lu2016, JKWoo}. \\

While statistical modeling and analyses for accommodating the various dependence structures in the collective risk model have been actively conducted, incorporating such a dependence structure in the BMS is largely overlooked in the literature. This lacuna could be attributed to the fact that developing and constructing a theoretically solid BMS structure that is determined by both frequency and severity is not straightforward once the frequency--severity relationship becomes dependent in complicated ways. An exception is the recent work of \citet{PengAhn}, where the authors show that the classical BMS can actually accommodate the frequency--severity dependence if we were to properly adjust the optimal BM relativity.
However, the transition rule in this BMS model is still completely governed by the frequency while ignoring the severity history; in this sense, the BMS proposed by \citet{PengAhn} is only a partial solution to the task of creating a more complete BMS that depends on both frequency and severity. We note that \citet{Frangos2001,Tzougas2014,Gomez, Lu2016, JKWoo} differently consider a posteriori ratemaking system with the past claim history of both frequency and severity, but the proposed premiums in these studies are written as mathematical formulas. Hence, the premiums are not easily convertible to a proper BMS with levels, transition rule, and relativity.\\

To this extent, we study a BMS that depends on the history of both frequency and severity, with the bivariate random effect model as the underlying statistical model. Specifically, we generalize the approach of \citet{PengAhn} and propose a new BMS transition rule where the past claim's size as well as the number of claims affect the next year's BM level.

Next, using the analogy between the Bayesian premium and the BMS premium, we analytically derive a set of optimal relativities that minimize the expected squared error in the presence of a frequency--severity dependence under this transition rule.

Furthermore, we show that, with a proper choice of the severity size threshold, the proposed BMS structure is guaranteed to outperform other frequency-driven BMSs for a statistical criterion. Given that current BMS practice neglects the severity history and is unable to penalize large-sized claims, the proposed method can provide a new, theoretically solid means of developing a BMS that is more faithful to the true premium or loss prediction.

The remaining paper is organized as follows. In section 2, we review the existing BMS along with the bivariate random effect model as the underlying statistical model. To motivate the use of the history of both frequency and severity, we consider a hypothetical frequency--severity model in section 3. In this model, the closed-form expressions are available for the Bayesian premium. Through this exercise, we illustrate how Bayesian premiums vary by the type of the incorporated claim history. Section 4 proposes a new BMS structure whose transition rule incorporates the history of both frequency and severity; we thus derive the optimal relativities analytically. In section 5, we show that the proposed BMS structure is guaranteed to outperform other frequency-driven BMS, with an optimally chosen severity size threshold; we then numerically illustrate this using a model inspired from actual dataset. A real insurance dataset is analyzed in section 6 to confirm the usefulness and superiority of the proposed BMS. Finally, section 7 concludes the paper.

\section{Existing BMS models: A review}\label{sec.3}
In a typical BM, the insurer classifies a given policyholder at the outset of the contract based on a priori risk characteristics that are observable. These characteristics would include, for example, gender, age, and vehicle type. As time passes, the policyholder's claim experiences are recorded, and the insurer will periodically reassess the risk of the policyholder and assign her to different risk classes. This so-called a posteriori risk classification based on the claim history involves the Bayesian premium; however, in practice, it is simplified and expressed as the BM relativity along with the transition rule.

Statistical modeling of a BMS is based on the collective risk model, where the parameters of the frequency and severity distributions vary over policyholders. These parameters are first estimated from the a priori risk characteristics at the moment of contract, and subsequently readjusted using the claim history over time, known as the posteriori risk classification process.\\

In this section, we review two existing BMSs. The first one is a classical system that aims for the predicted frequency, ignoring the severity component. The second BMS is a more recent one and targets the predicted aggregate severity based on the joint modeling of frequency and severity. However, the posteriori risk classifications in these BMSs are determined solely based on the frequency history under both BMSs.\\

To fix the notation, we denote $\mathbb{N}$, $\mathbb{N}_0$, $\mathbb{R}$, and $\mathbb{R}^+$ for the set of natural numbers, non-negative integers, real numbers, and positive real numbers, respectively. For a given policyholder, define $N_{t}$ to indicate the random frequency in the $t$th policy year, that is, the time interval between time $t-1$ and $t$. Use
\[
\boldsymbol{Y}_{t}=\begin{cases}
 (Y_{t, 1}, \cdots, Y_{t, N_{t}}) & N_{t}>0\\
 \hbox{Not defined} & N_{t}=0
\end{cases}
\]
to denote the associated individual severities. We further define the aggregate and average severity as

\[
S_{t}=\begin{cases}
\sum\limits_{j=1}^{N_{t}}Y_{t,j} & N_{t}>0\\
0 &N_{t}=0
\end{cases}
\quad\quad
\hbox{and}
\quad\quad
M_{t}=\begin{cases}
\frac{\sum\limits_{j=1}^{N_{t}}Y_{t,j}}{N_{t}} & N_{t}>0\\
\hbox{0} &N_{t}=0
\end{cases},
\]
respectively, so that we have the following relationship between $S_t$ and $M_t$:
\[
S_{t}=N_{t}M_{t}.
\]
Realizations of $N_{t}$, $\boldsymbol{Y}_{t}$, $S_{t}$, and $M_{t}$ are denoted by $n_{t}$, $\boldsymbol{y}_{t}$, $s_{t}$, and $m_{t}$, respectively. Furthermore, we define the claim history of the given policyholder at the end of year $T$ as
\[
\mathcal{F}_{T}^{\rm [full]}=\left\{(n_{t}, \boldsymbol{y}_{t} )\big\vert t=1, \cdots, T \right\}\quad\hbox{and}\quad
\mathcal{F}_{T}^{\rm [freq]}=\left\{n_{t}\big\vert t=1, \cdots, T \right\},
\] where the latter contains the frequency information only.\\

As is common in the BMS literature, we differentiate the \textit{a priori} and \textit{a posterior} risk classification. The former is based on the risk characteristics of policyholders (or explanatory variables), which are observed at the moment of contract, and the latter is based on a residual effect, which cannot be explained by the a priori risk characteristics. For our purposes, we like to differentiate these two for the frequency and severity components, separately. Thus, we write
\[
\boldsymbol{X}^{[1]}=\left(X_{1}^{[1]}, \cdots, X_{d_1}^{[1]}\right)\quad\hbox{and}\quad\Theta^{[1]}
\]
to denote the a priori risk characteristics and the residual effect for the frequency and
\[
\boldsymbol{X}^{[2]}=\left(X_{1}^{[2]}, \cdots, X_{d_2}^{[2]}\right)\quad\hbox{and}\quad\Theta^{[2]}
\]
to denote the same quantities for the severity for the given policyholder. Superscripts $[1]$ and $[2]$ will be extensively used throughout the paper. As done in the literature, the a priori risk characteristics---$\boldsymbol{X}^{[1]}$ and $ \boldsymbol{X}^{[2]}$---are constant in time and may overlap; they share the same explanatory variables. Later, we will sometimes use an additional subscript $i$ for all the quantities introduced above to represent the $i$th policyholder, as needed. Regarding the a priori risk classification, the insurer is assumed to predetermine $\mathcal{K}$ risk classes based on the policyholders' risk characteristics. We will let $\left(\boldsymbol{x}_{\kappa}^{[1]}, \boldsymbol{x}_{\kappa}^{[2]}\right)$ define the $\kappa$th risk class and $w_\kappa$ be the weight of the risk class:
 \begin{equation}\label{eq.1}
 w_\kappa= \P{\boldsymbol{X}^{[1]}=\boldsymbol{x}_{\kappa}^{[1]},
 \boldsymbol{X}^{[2]}=\boldsymbol{x}_{\kappa}^{[2]}
 }, \quad \kappa = 1,\cdots, \mathcal{K}.
 \end{equation}
and sometimes it is convenient to define the frequency-only case
\[
w_\kappa^{[1]}=\P{\boldsymbol{X}^{[1]}=\boldsymbol{x}_{\kappa}^{[1]}}, \quad \kappa = 1,\cdots, \mathcal{K},
\]

For the statistical modeling of severity, we consider the {\it(reproductive) exponential dispersion family} (EDF) in \citet{Nelder1989}. The EDF is denoted by ${\rm ED}(\mu, \psi)$, where $\mu$ and $\psi$ represent the mean and dispersion parameter, respectively. It has the probability density/mass function of the form
\[
p(y\big\vert \theta,\psi)=\exp\left[(y\theta-b(\theta))/\psi \right]c(y,\psi),
\]
where $b(\cdot)$ and $c(\cdot)$ are predetermined functions, and $\theta$ is the canonical parameter. With such parametrization, the mean and variance are given as
$$\mu=\E{Y}=b^\prime(\theta),$$
and
$$\Var{Y}=b^{\prime\prime}(\theta)\psi\equiv V(\mu)\psi,$$
for some variance function $V(\cdot)$.
The inverse of $b^{'}(\cdot)$ is known as the link function, often denoted by $\eta(\cdot)$. One important property of the EDF used in this paper is the reproductive property, that is, for $Y_1, \cdots, Y_n\iid {\rm ED}(\mu, \psi)$, the sample average is distributed as
\[
\frac{1}{n}\sum\limits_{j=1}^{n}Y_j\sim {\rm ED}\left(\mu, \frac{\psi}{n}\right).
\]

\subsection{Frequency-based BMS}\label{sec.freq-only.BMS} 

In the classical BMS, both the a priori and posteriori classification are fully determined by the frequency while ignoring the severity. We present a statistical random effect model for this BMS, calling it ``the frequency model.'' For a given policyholder, the frequency model consists of several components, described as follows: \\

The frequency, $N_{t}$, of a given policyholder is specified using a count regression model conditioning on $\boldsymbol{X}^{[1]}=\boldsymbol{x}_\kappa^{[1]}$
 and $\Theta^{[1]}=\theta^{[1]}$, that is
 \begin{equation}\label{eq.402}
 N_{t}\big\vert \left(\boldsymbol{x}_\kappa^{[1]},\theta^{[1]}\right) \iid F_1\left(\cdot; \lambda_\kappa^{[1]}\theta^{[1]}, \psi^{[1]} \right), 
 \end{equation}
 for some discrete distribution $F_1$ with the mean parameter $\lambda_\kappa^{[1]}\theta^{[1]}$ and the other parameter $\psi^{[1]}$. Following the convention, the a priori rate is modeled as $\lambda_\kappa^{[1]}= \left(\eta^{[1]}\right)^{-1}\left(\boldsymbol{x}_\kappa^{[1]}\boldsymbol{\beta}^{[1]}\right)$, where $\boldsymbol{\beta}^{[1]}$ is the regression coefficient vector to be estimated from the data. Hence, having information on $\boldsymbol{x}_\kappa^{[1]}$ is equivalent to knowing $\lambda_\kappa^{[1]}$, provided that the model has been fully specified. For the brevity of exposition, we will use $\lambda_\kappa^{[1]}$ instead of $\boldsymbol{x}_\kappa^{[1]}$ when the a priori characteristics $\boldsymbol{x}_\kappa^{[1]}$ appears as a conditional quantity; for example, $\E{N_{T+1}\big\vert \lambda_\kappa^{[1]}, \theta^{[1]}}$ will be used instead of $\E{N_{T+1}\big\vert \boldsymbol{x}_\kappa^{[1]}, \theta^{[1]}}$.
 Note that $\lambda_\kappa^{[1]}$ and $\theta^{[1]}$ vary over policyholders, whereas $\psi^{[1]}$ remains constant, common for all policyholders. As such, we will suppress this common parameter for quantities that do not vary over individual policyholders. In a special case where $F_1 \in$ EDF, $\psi^{[1]}$ becomes the common dispersion parameter.\\

By definition, the observations $N_t$ are conditionally independent given the a priori risk characteristics $\boldsymbol{x}_\kappa^{[1]}$ and residual (or random) effect $\theta^{[1]}$, but are dependent unconditionally.
We assume that $\Theta^{[1]}$ is distributed as
 \begin{equation}\label{eq.403}
 \Theta^{[1]}\iid G_1,
 \end{equation}
 where $G_1$ is a proper distribution function with fixed parameters.
Individual policyholders have different realizations of $\Theta^{[1]}$, which capture the residual heterogeneity; it is assumed that $ \E{\Theta^{[1]}} =1 $ for convenience, following the convention.\\

The ideal premium for the next year in this frequency model is the hypothetical mean, also known as individual premium in the credibility literature, which is written as
\[\E{N_{T+1}\big\vert \lambda_\kappa^{[1]}, \theta^{[1]}}.\]
This, however, can be numerically obtained only when $(\Lambda^{[1]}, \Theta^{[1]})$ are known.
While $\Lambda^{[1]}$ is known at the outset of the contract and remains fixed over time, the value of $\Theta^{[1]}$ is never known, as it is a random variable. The solution is then to revise the distribution parameters of $\Theta^{[1]}$ based on the past claim frequency history $\mathcal{F}_{T}^{\rm [freq]}$ and compute the predictive mean $\E{N_{T+1}\big\vert \lambda_\kappa^{[1]}, \mathcal{F}_{T}^{\rm [freq]} }$. This is called the Bayesian premium.\\

Although the Bayesian premium is deemed to be the best premium possible, most ratemaking systems in the auto insurance industry maintain a BMS, which can be considered a discrete approximation of the Bayesian premium, as explained below.
The classical BMS constitutes three elements: the BMS levels, transition rule, and relativities. Without loss of generality, let us assume that the insurer has $z+1$ BM levels varying from $0$ to $z$, with $-1/+h$ system as a transition rule. Thus, in this system, a claim-free year is rewarded by moving one BM level downward, while each reported claim, regardless of the claim size, is penalized by climbing $h$ levels per claim at the end of the period. Theoretically, the choice of the number of BM levels and the transition rule are arbitrary, although insurers may have constraints subject to local regulations in practice. The transition rule is important in formal analyses of any BMS, since it affects how the drivers are populated over $z+1$ different BM levels, which, in turn, determines the optimal premium of a driver in each level for the following year. \\

The last element of the classical BMS is the relativity. In the classical frequency-driven BMS, the relativity associated with BM level $\ell$, denoted by ${r}^{\rm [freq]}(\ell)$, is used to predict the frequency in the next year through the multiplicative form
\begin{equation}\label{eq.b3}
\lambda_\kappa^{[1]}r^{\rm [freq]}(\ell), \qquad \ell=0,...,z,
\end{equation}
where the a priori rate
 $\lambda_{\kappa}^{[1]}= \left(\eta^{[1]}\right)^{-1}\left(\boldsymbol{x}_{\kappa}^{[1]}\boldsymbol{\beta}^{[1]}\right)$ reflects the risk class to which the policyholder belongs.
The expression (\ref{eq.b3}) is motivated from the following minimization by \citet{Chong}:
\begin{equation}\label{eq.h.5}
\begin{aligned}
(\tilde{r}^{\rm [freq]}(0), \cdots, \tilde{r}^{\rm [freq]}(z))&=\argmin_{({r}(0), \cdots, {r}(z))\in \mathbb{R}^{{z+1}}} \E{\left(
\E{N_{T+1}\big\vert \Lambda^{[1]}, \Theta^{[1]}}-\Lambda^{[1]}{r}(L)\right)^2}\\
&=\argmin_{({r}(0), \cdots, {r}(z))\in \mathbb{R}^{{z+1}}} \sum\limits_{\ell=0}^{z}\E{\left(
\Lambda^{[1]}\Theta^{[1]} -\Lambda^{[1]}{r}(L)\right)^2|L=\ell}\P{L=\ell},\\
\end{aligned}
\end{equation}
whose solution, or optimal relativity, is given by
\begin{equation}\label{eq.60}
\tilde{r}^{\rm [freq]}(\ell)=\frac{\E{ \left( \Lambda^{[1]} \right)^2 \Theta^{[1]} \big\vert L=\ell}}{\E{\left( \Lambda^{[1]} \right)^2\big\vert L=\ell}}
=\frac{ \sum_{\kappa \in \mathcal{K}} w_\kappa^{[1]} \left(\lambda_\kappa^{[1]}\right)^2\, \int \theta^{[1]}\, \pi_\ell(\lambda_\kappa^{[1]}\theta^{[1]}) g_1(\theta^{[1]}){\rm d}\theta^{[1]} }
{\sum_{\kappa \in \mathcal{K}} w_\kappa^{[1]} \left(\lambda_\kappa^{[1]}\right)^2\, \int \pi_\ell(\lambda_\kappa^{[1]}\theta^{[1]})g(\theta^{[1]}){\rm d}\theta^{[1]} }.
\end{equation}

Thus, for a driver currently in BM level $\ell$, one can use $\lambda_\kappa^{[1]}\tilde{r}^{\rm [freq]}(\ell)$ as an optimal predictor of the next year's frequency. We have two comments on this result: First, the optimal set of relativities are designed to minimize the expected squared error of the individual premium $\E{N_{T+1}\big\vert \lambda_\kappa^{[1]}, \theta^{[1]}}$. In this sense, the values $\lambda_\kappa^{[1]}\tilde{r}^{\rm [freq]}(\ell)$, $\ell=0,...,z$, can be considered an approximated prediction of the individual premium via $z+1$ categories. Furthermore, for a specific policyholder, if we incorporate her claim history and let $\Theta^{[1]}$ be revised accordingly, the resulting $\lambda_\kappa^{[1]}\tilde{r}^{\rm [freq]}(\ell)$, $\ell=0,...,z$, values can also be seen as an approximation of the Bayesian premium that fluctuates over time. That is, we can interpret $\tilde{r}^{\rm [freq]}(\ell)$ as the adjustment term based on the claim history because the BMS premium (\ref{eq.b3}) is analogous to the following Bayesian premium:
\begin{equation}\label{eq.a1}
\begin{aligned}
\E{N_{T+1}\big\vert \lambda_\kappa^{[1]}, \mathcal{F}_{T}^{\rm [freq]} }
&=\lambda_\kappa^{[1]}\E{\Theta^{[1]}\big\vert \lambda_\kappa^{[1]}, \mathcal{F}_{T}^{\rm [freq]} }.
\end{aligned}
\end{equation}
Second, we note that the optimal relativity depends on the form of the object function to minimize; the frequency model as shown in (\ref{eq.h.5}) targets the hypothetical frequency mean $\E{N_{T+1}\big\vert \Lambda^{[1]}, \Theta^{[1]}}$. Later, we will consider different object functions.\\

Before closing this section, we briefly discuss the Markov chain embedded in this classical BMS. Let us denote the transition probability of a policyholder with an expected frequency $\lambda_\kappa^{[1]}\theta^{[1]}$
from BM level $\ell$ in the current year to BM level $\ell^*$ in the next year under $-1/+h$ transition rule by
 \begin{equation}\label{eq.c1}
 p_{\ell, \ell^*}(\lambda_{\kappa}^{[1]}\theta^{[1]})
 \end{equation}
 using the realizations of the parameters associated with the frequency.
As the BM level in the next period is completely characterized by
the current BM level and the number of reported claims in the current period, the BMS can be naturally represented by a Markov chain whose transition matrix elements are given in (\ref{eq.c1}). For the given a priori rate $\lambda_\kappa^{[1]}$ and the residual effect $\theta^{[1]}$, we may denote the corresponding stationary distribution as
\begin{equation}\label{eq.h.2}
\boldsymbol{\pi}\left(\lambda_\kappa^{[1]}\theta^{[1]}\right)=\left(
\pi_0\left(\lambda_\kappa^{[1]}\theta^{[1]}\right), \cdots, \pi_z\left(\lambda_\kappa^{[1]}\theta^{[1]}\right)
\right)^T,
\end{equation}
where $\pi_\ell\left(\lambda_\kappa^{[1]}\theta^{[1]}\right)$ is the stationary distribution for the policyholder with expected frequency $\lambda_\kappa^{[1]}\theta^{[1]}$ to be in level $\ell$ induced by the transition probability in \eqref{eq.c1}.
When marginalized, we can find the unconditional stationary distribution as well. That is, if we use $L$ to denote the BM level occupied by a randomly selected policyholder in the steady state, the distribution of $L$ is expressed as
\begin{equation*}
\begin{aligned}
\P{L=\ell}
&=\sum\limits_{\kappa \in\mathcal{K}} w_{\kappa} \int \pi_\ell\left(\lambda_{\kappa}^{[1]}\theta \right) g_1(\theta){\rm d}\theta, \quad \hbox{for\quad$\ell=0, \cdots, {z}$},
\end{aligned}
\end{equation*}
where $g_1$ is the density function of $G_1$. This probability is used in the minimization (\ref{eq.h.5}).
Notably, as shown in the above equation, the distribution of $L$ depends on various factors including the distributional assumptions for the conditional distribution of frequency and random effect, transition rule, and stationary distribution.

\subsection{BMS based on both frequency and severity}\label{sec.freq-sev.BMS}

Though widely used in practice, the frequency-based BMS is less than optimal: It completely ignores the claim severity information. A more statistically solid approach would be based on the aggregate severity, or the total loss of the collective risk model, which reflects both claim frequency and severity.
In traditional collective risk models, the frequency and severity are assumed independent. Although, a series of recent studies shows that this unwarranted assumption can be relaxed in the collective risk model, and, thus, further in the Bayesian premium and BMS. In fact, the dependence between the frequency and severity in auto insurance is shown to be statistically significant in recent empirical studies such as \citet{Frees4} and \citet{Garrido}. \\

 While various dependent collective risk models exist as mentioned in Introduction, including the two-step frequency--severity model and the copula-based models, we find that the bivariate random effect model is most convenient for the development of a better BMS because of its conditioning structure in the model specification. In this section, we focus on the model by \citet{PengAhn}, as it offers the most comprehensive bivariate random effect model: It contains other existing bivariate random effect models in the literature---\citet{Pinquet, Cossette, Bastida, Czado2015, Lu2016, JKWoo}---as special cases.\\

Paralleling the frequency model in section \ref{sec.freq-only.BMS}, the dependent collective risk model of \citet{PengAhn} is specified as follows. We call this ``the frequency--severity model,'' as it requires both frequency and severity models. For the given a priori rate $(\Lambda^{[1]}, \Lambda^{[2]})=(\lambda_\kappa^{[1]}, \lambda_\kappa^{[2]})$ and residual effect $(\Theta^{[1]}, \Theta^{[2]})=(\theta^{[1]}, \theta^{[2]})$, the observations
 \[
 N_{t}, \quad \hbox{and}\quad Y_{t,j}\quad\hbox{for}\quad t,j\in\mathbb{N}
 \]
 are independent. In the frequency--severity model, the distributional assumption for $N_t$ is identical to the frequency model in section \ref{sec.freq-only.BMS}.
 The individual severity $Y_{t,j}$ is specified using a regression model conditioning on the a priori risk characteristics $\boldsymbol{X}^{[2]}=\boldsymbol{x}_\kappa^{[2]}$, or equivalently $ \Lambda^{[2]}=\lambda_\kappa^{[2]}$ as previously mentioned, and the residual effect $\Theta^{[2]}$, such that
 \begin{equation}\label{eq.y}
 Y_{t,j}\big\vert \left( \lambda_\kappa^{[2]}, \theta^{[2]}\right) \iid {\rm EDF}\left( \lambda_\kappa^{[2]}\theta^{[2]}, \psi^{[2]}\right),
 \end{equation}
 where the EDF has mean parameter $\lambda_\kappa^{[2]}\theta^{[2]}$ with $\lambda_\kappa^{[2]}=\left(\eta^{[2]}\right)^{-1}\left(\boldsymbol{x}_\kappa^{[2]} \boldsymbol{\beta}^{[2]}\right)$
 and dispersion parameter $\psi^{[2]}$.
 Finally, the residual effect characteristic $\left(\Theta^{[1]},\Theta^{[2]}\right)$, which is assumed to be independent of the a priori rate $(\Lambda^{[1]}, \Lambda^{[2]})$,
 has a joint distribution $H$. That is,
 \begin{equation}\label{eq.4}
 \left(\Theta^{[1]}, \Theta^{[2]} \right)\iid H=C\left(G_1, G_2 \right),
 \end{equation}
 where $G_1$ and $G_2$ denote the marginal distribution functions for $\Theta^{[1]}$ and $\Theta^{[2]}$, respectively, and $C$ is a bivariate copula that defines the dependence structure in $H$.
 We use $g_1$, $g_2$, and $h$ to denote the density functions of $G_1$, $G_2$, and $H$, respectively.
As before, it is further assumed
\[
 \E{\Theta^{[1]}} = \E{\Theta^{[2]}} = 1
 \]
for convenience.\\

There are several conditional independence relations under this frequency--severity model specification. For example, given $(\lambda_\kappa^{[1]}, \lambda_\kappa^{[2]})$, and $(\theta^{[1]}, \theta^{[2]})$, the claim history
 $ \left(N_{t}, \boldsymbol{Y}_{t}\right) $ of a given policyholder
 is independent across time, and independence between frequency and individual severities and independence among individual severities are further implied.
However, marginally or unconditionally, various dependence relations are observed:
\begin{itemize}
 \item Dependence among frequencies induced by residual effect $\Theta^{[1]}$;
 \item Dependence among severities induced by residual effect $\Theta^{[2]}$;
 \item Dependence between frequency and severity induced by residual effects $(\Theta^{[1]},\Theta^{[2]}$) via the copula $C$.
\end{itemize}

Recalling that the individual severity $Y_{t,j}$ is assumed EDF distributed, we note that the conditional average severity is again EDF distributed. That is,
 \begin{equation}\label{eq.m}
 M_{t}\big\vert \left(\lambda_\kappa^{[1]}, \lambda_\kappa^{[2]},\theta^{[1]},\theta^{[2]}, N_{t}\right) \iid {\rm ED}\left( \lambda_\kappa^{[2]}\theta^{[2]}, \, \psi^{[2]}/N_{t}\right), \quad N_{t}>0,
 \end{equation}
 and
 \[
 \P{M_{t}=0\big\vert \lambda_\kappa^{[2]},\theta^{[2]}, N_{t} }=1, \quad \text{ if }N_{t}=0,
 \]
 where the mean parameter of the EDF, $\lambda_\kappa^{[2]}\theta^{[2]}$, is modeled through $\lambda_\kappa^{[2]}=\left(\eta^{[2]}\right)^{-1}\left(\boldsymbol{x}_\kappa^{[2]} \boldsymbol{\beta}^{[2]}\right)$. Thus, the distribution of $M_t$ is the same as that of $Y_{t,j}$, except for the dispersion parameter. This allows us to efficiently model the aggregate severity $S_t$ as a product form $N_t M_t$, whose mean serves as the premium under the frequency--severity model. There are two noteworthy sub-modes of the frequency--severity model. First, it reduces to the frequency model in section \ref{sec.freq-only.BMS} by removing the severity part. Second, if the copula $C$ embedded in $H$ is an independent copula, then the frequency and severity can be independently modeled, leading to the standard collective risk model.\\

We found that the optimal frequency-premium takes a form of $\lambda_\kappa^{[1]}\tilde{r}^{\rm [freq]}(\ell)$ in the frequency model from (\ref{eq.b3}). Now, suppose we predict the aggregate severity as the premium using the history of the frequency only. A simple, but widely, used approach in the literature \citep{Denuit2} is to multiply
the average severity of all policyholders in the same risk class to the frequency such that the BMS premium is predicted with
\begin{equation}
\label{eq.h.9}
\lambda_{\kappa}^{[1]}\tilde{r}^{\rm [freq]}(\ell) \, \times \, \lambda_{\kappa}^{[2]}.
\end{equation}
This multiplicative form is less than ideal in at least two aspects. First, by incorporating the average severity of the entire portfolio, it reflects the history of the frequency of a given policyholder, but completely ignores the past severity records of individual policyholders. Second, the multiplicative form can be justified only when the independence of the frequency and severity is warranted, which is not valid in general. In fact, the BMS premium (\ref{eq.h.9}) is motivated from the following Bayesian premium with the implied independence between frequency and severity:
\begin{equation}\label{eq.h.8}
\begin{aligned}
&\E{S_{T+1}\big\vert \lambda_\kappa^{[1]}, \lambda_\kappa^{[2]}, \mathcal{F}_{T}^{\rm [freq]} }\\
&
=\E{N_{T+1}
\E{M_{T+1} \big\vert \lambda_\kappa^{[1]}, \lambda_\kappa^{[2]}, \Theta^{[1]}, \Theta^{[2]}, \mathcal{F}_{T}^{\rm [freq]}, N_{T+1} }
\big\vert \lambda_\kappa^{[1]}, \lambda_\kappa^{[2]}, \mathcal{F}_{T}^{\rm [freq]} }\\
&=\E{N_{T+1}\Theta^{[2]}
\big\vert \lambda_\kappa^{[1]}, \lambda_\kappa^{[2]}, \mathcal{F}_{T}^{\rm [freq]} }\times \lambda_\kappa^{[2]}\\
&=\lambda_\kappa^{[1]}\E{\Theta^{[1]}\Theta^{[2]}\big\vert \lambda_\kappa^{[1]}, \lambda_\kappa^{[2]}, \mathcal{F}_{T}^{\rm [freq]} }\times \lambda_\kappa^{[2]}\\
&=\lambda_\kappa^{[1]}\E{\Theta^{[1]}
\big\vert \lambda_\kappa^{[1]}, \lambda_\kappa^{[2]}, \mathcal{F}_{T}^{\rm [freq]}}
\times \lambda_\kappa^{[2]}\E{\Theta^{[2]}\big\vert \lambda_\kappa^{[1]}, \lambda_\kappa^{[2]}, \mathcal{F}_{T}^{\rm [freq]} }
\\
&=\lambda_\kappa^{[1]}\E{\Theta^{[1]}\big\vert \lambda_\kappa^{[1]}, \lambda_\kappa^{[2]}, \mathcal{F}_{T}^{\rm [freq]} } \, \times \, \lambda_\kappa^{[2]},
\end{aligned}
\end{equation}
where the first and third equalities are from the law of total expectation and the fourth and fifth equalities are from the independence between $\Theta^{[1]}$ and $\Theta^{[2]}$.
From this, the relativity $\tilde{r}^{\rm [freq]}(\ell)$ in \eqref{eq.h.9} can be interpreted as the BMS version of the posterior estimation
\[\E{\Theta^{[1]}\big\vert \lambda_\kappa^{[1]}, \lambda_\kappa^{[2]}, \mathcal{F}_{T}^{\rm [freq]} }\]
in \eqref{eq.h.8}.
Therefore, while convenient and simple, the BMS premium \eqref{eq.h.9} may result in a considerable bias when dependence between frequency and severity is strong. To adjust such bias, \citet{PengAhn} suggest the following BMS premium from for a policyholder at BM level $\ell$ in a $-1/+h$ system under the frequency--severity model:
\begin{equation}\label{eq.h.10}
\lambda_{\kappa}^{[1]}\lambda_{\kappa}^{[2]} \tilde{r}^{\rm [dep]}(\ell), \qquad \ell=0,\cdots,z,
\end{equation}
where optimal relativity $\tilde{r}^{\rm [dep]}(\ell)$ for $\ell=0,\cdots,z$
is the solution of the optimization,
\begin{equation}\label{eq.h.11}
\begin{aligned}
(\tilde{r}^{\rm [dep]}(0), \cdots, \tilde{r}^{\rm [dep]}(z))&=\argmin_{({r}(0), \cdots, {r}(z))\in \mathbb{R}^{{z+1}}} \E{\left(
\E{S_{T+1}\big\vert \Lambda^{[1]}, \Lambda^{[2]}, \Theta^{[1]}, \Theta^{[2]} }-\Lambda^{[1]}\Lambda^{[2]}{r}(L)\right)^2}\\
&=\argmin_{({r}(0), \cdots, {r}(z))\in \mathbb{R}^{{z+1}}}
\sum\limits_{\ell=0}^{z}\E{\left(
\Lambda^{[1]}\Lambda^{[2]}\Theta^{[1]}\Theta^{[2]} -\Lambda^{[1]}\Lambda^{[2]}{r}(L)\right)^2|L=\ell}\P{L=\ell},\\
\end{aligned}
\end{equation} which is an extended setting of that in \citet{Chong}.
The solution of \eqref{eq.h.11} can be obtained as
 \begin{equation}\label{eq.160}
\tilde{r}^{\rm [dep]}({l})=\frac{\E{\left( \Lambda^{[1]}\Lambda^{[2]}\right)^2 \Theta^{[1]}\Theta^{[2]}\bigg\vert L=\ell}}{\E{\left(\Lambda^{[1]} \Lambda^{[2]}\right)^2\big\vert L=\ell}}
\quad\hbox{for}\quad \ell=0, \cdots,{z},
\end{equation}
where the numerator and denominator can be analytically expressed as
\begin{equation}\label{eq.158}
\E{\left( \Lambda^{[1]}\Lambda^{[2]}\right)^2 \Theta^{[1]}\Theta^{[2]}\bigg\vert L=\ell}=\frac{\sum\limits_{\kappa\in\mathcal{K}} w_{\kappa} \left(\lambda_{\kappa}^{[1]}\lambda_{\kappa}^{[2]}\right)^2 \int\int \theta^{[1]}\theta^{[2]} \pi_\ell(\lambda_{\kappa}^{[1]}\theta^{[1]}) h(\theta^{[1]}, \theta^{[2]}){\rm d}\theta^{[1]}{\rm d} \theta^{[2]}}
{\sum\limits_{\kappa\in\mathcal{K}}w_{\kappa} \int \pi_\ell(\lambda_{\kappa}^{[1]}\theta^{[1]}) g_1(\theta^{[1]}){\rm d}\theta^{[1]} }
\end{equation}
and
\begin{equation}\label{eq.159}
\E{\left( \Lambda^{[1]}\Lambda^{[2]}\right)^2\big\vert L=\ell} = \frac{\sum\limits_{\kappa\in\mathcal{K}} w_{\kappa} \left(\lambda_{\kappa}^{[1]}\lambda_{\kappa}^{[2]}\right)^2 \int \pi_\ell(\lambda_{\kappa}^{[1]}\theta^{[1]}, \psi^{[1]})g_1(\theta^{[1]}){\rm d}\theta^{[1]}}
{\sum\limits_{\kappa\in\mathcal{K}} w_{\kappa} \int \pi_\ell(\lambda_{\kappa}^{[1]}\theta^{[1]}, \psi^{[1]}) g_1(\theta^{[1]}){\rm d}\theta^{[1]} },
\end{equation}
respectively, where $\pi_\ell$ is a stationary distribution defined in \eqref{eq.h.2}.
 We refer to \citet{PengAhn} for the detailed derivation of \eqref{eq.160}.\\

 Similar to \eqref{eq.h.9}, the premium \eqref{eq.h.10} can be viewed as the BMS version of the Bayesian premium from
 \begin{equation}\label{kim.1}
\begin{aligned}
&\E{S_{T+1}\big\vert \lambda_\kappa^{[1]}, \lambda_\kappa^{[2]}, \mathcal{F}_{T}^{\rm [freq]} }\\
&
=\E{N_{T+1}
\E{M_{T+1} \big\vert \lambda_\kappa^{[1]}, \lambda_\kappa^{[2]}, \Theta^{[1]}, \Theta^{[2]}, \mathcal{F}_{T}^{\rm [freq]}, N_{T+1} }
\big\vert \lambda_\kappa^{[1]}, \lambda_\kappa^{[2]}, \mathcal{F}_{T}^{\rm [freq]} }\\
&= \lambda_\kappa^{[2]} \E{N_{T+1}\Theta^{[2]}
\big\vert \lambda_\kappa^{[1]}, \lambda_\kappa^{[2]}, \mathcal{F}_{T}^{\rm [freq]} }\\
&=\lambda_\kappa^{[1]}\lambda_\kappa^{[2]}
\E{\Theta^{[1]}\Theta^{[2]}\big\vert \lambda_\kappa^{[1]}, \lambda_\kappa^{[2]}, \mathcal{F}_{T}^{\rm [freq]} }\\
\end{aligned}
\end{equation}
and, consequently, the relativity $\tilde{r}^{\rm [dep]}(\ell)$ in \eqref{eq.h.10} can be interpreted as the BMS version of the posterior estimation

\begin{equation}
\label{kim.11}
\E{\Theta^{[1]}\Theta^{[2]}\big\vert \lambda_\kappa^{[1]}, \lambda_\kappa^{[2]}, \mathcal{F}_{T}^{\rm [freq]} }.
\end{equation}
Notably, the BMS premium of \citet{PengAhn} operates under the $-1/+h$ system, where the BM level transition is driven by the claim frequency only and is unable to penalize high severity claims. Also, as seen from (\ref{kim.11}), the severity history is not used inside the condition. In reality, the severity information is indirectly incorporated only through the dependence between $\Theta^{[1]}$ and $\Theta^{[2]}$. Thus, the BMS premium of \citet{PengAhn} approximates the Bayesian premium, but in a limited sense, because it incorporates the history of claim frequency and not severity. Later, we propose to improve the BMS premium prediction with a direct incorporation of the history of severity as well.

\subsection{Comparative summary so far}
To provide a broader view, it is informative to compare the BMS models considered in this section. In the following, the BMS models are listed in an increasing order with later models being more sophisticated.
\begin{enumerate}
 \item The frequency-only BMS in section \ref{sec.freq-only.BMS}. This simplest BMS mainly targets the prediction of frequency. At the outset of the contract, a priori risk characteristics $\boldsymbol{x}_\kappa^{[1]}$ are used to calibrate the a priori rate, and the policyholder's specific residual effect $\Theta^{[1]}$ is periodically readjusted based on the claim frequency history $\mathcal{F}_{T}^{\rm [freq]}$. The optimal BMS predictor of the next year's frequency is $\lambda_\kappa^{[1]}\tilde{r}^{\rm [freq]}(\ell)$ in
 (\ref{eq.b3}).
 \item The simple BMS predicting the aggregate severity under the independence between frequency and severity, or equivalently the independence $\Theta^{[1]}$ and $\Theta^{[2]}$, under the frequency--severity model in section \ref{sec.freq-sev.BMS}.
 In this model, the estimator of $\Theta^{[1]}$ is periodically readjusted based on the claim frequency history $\mathcal{F}_{T}^{\rm [freq]}$. However, the severity history is ignored, which implicitly leads the predictor of $\Theta^{[2]}$ being equal to its mean, 1, a fixed constant. The optimal BMS predictor of the next year's premium is
$\lambda_{\kappa}^{[1]} \lambda_{\kappa}^{[2]}\tilde{r}^{\rm [freq]}(\ell)$ because of the independence between frequency and severity, as shown in (\ref{eq.h.9}).

 \item The BMS predicting the aggregate severity with dependence between frequency and severity, proposed by \citet{PengAhn}. This model extends the previous model; the optimal BMS takes the form $\lambda_{\kappa}^{[1]}\lambda_{\kappa}^{[2]} \tilde{r}^{\rm [dep]}(\ell)$ as shown in (\ref{eq.h.10}), where $\tilde{r}^{\rm [dep]}(\ell)$ is the relativity that accommodates the dependence between $\Theta^{[1]}$ and $\Theta^{[2]}$. In terms of the Bayesian premium framework, in this model, the predictor of $\Theta^{[1]}$ is periodically readjusted based on $\mathcal{F}_{T}^{\rm [freq]}$. The revised $\Theta^{[1]}$, in turn, affects $\Theta^{[2]}$ through the dependence structure in $H$, the joint distribution of $( \Theta^{[1]},\Theta^{[2]})$. The severity history is, however, still ignored, as the transition rule is purely driven by frequency.

\end{enumerate}

\section{Motivation for transition rules driven by severity history}\label{sec.4}
In this section, we investigate a hypothetical frequency--severity model where closed-form expressions are available for various quantities. By analyzing this model analytically, we derive various Bayesian premiums in the presence of claim history. In particular, we show that Bayesian premiums vary not only by the target variable (i.e., the frequency or the aggregate severity), but also the type of claim history incorporated.\\

In light of the analogy between the Bayesian premium and the BMS premium, a major implication of this exercise is the need for a BMS that has a transition rule operated by the past claim history of both frequency and severity. We thus conclude that, while the optimal relativities in \eqref{eq.h.11} provide the best prediction under the BMS, where the transition is only ruled by claim frequency, considering severity in the transition rule might improve the prediction ability of the BMS.

\subsection{The model}
The model we consider is a frequency--severity model with the distribution of frequency given by
\[
 N_{t}\big\vert \left(\lambda_\kappa^{[1]},\lambda_\kappa^{[2]},\theta^{[1]},\theta^{[2]}\right) \iid {\rm Pois}\left( \lambda_\kappa^{[1]}\theta^{[1]}\right), 
\]
and the individual severities by
\begin{equation}\label{eq.21}
 Y_{t,j}\big\vert \left(\lambda_\kappa^{[1]},\lambda_\kappa^{[2]},\theta^{[1]},\theta^{[2]}\right) \iid {\rm Pois}\left( \lambda_\kappa^{[2]}\theta^{[2]}\right).
\end{equation}
 For the residual effect characteristics, the following joint density of $(\Theta^{[1]}, \Theta^{[2]} )$ is used
\begin{equation}\label{eq.a.1}
h(\theta_1,\theta_2)=c_0 c_1\exp(-c_1\theta_1)c_1\exp(-c_1\theta_2) +(1-c_0)c_2\exp(-c_2\theta_1)c_2\exp(-c_2\theta_2)
\end{equation}
whose corresponding marginal distributions $G_1$ and $G_2$ are a mixture of exponential distribution with the identical density
\[
g_1(\theta)=g_2(\theta)=c_0c_1\exp(-c_1\theta) + (1-c_0)c_2\exp(-c_2\theta), \quad c_1, c_2>0\quad\hbox{and}\quad c_0\in[0,1].
\]
A similar joint distribution, based on conjugacy, appears in \citet{Lu2016} and \citet{JKWoo}. The independence between $\Theta^{[1]}$ and $\Theta^{[2]}$ corresponds to the boundary values of $c_0$ range, that is, $c_0=0$ or $1$, in which case the joint density (\ref{eq.a.1}) reduces to the product of two exponential densities. For $0<c_0<1$, the dependence is positive. It is further assumed that
\[
c_0\frac{1}{c_1}+(1-c_0)\frac{1}{c_2}=1
\]
such that the mean of each marginal distribution is 1. Finally, for the convenience of interpretation, we assume $c_1>c_2$ such that $\E{\Theta^{[1]}}<\E{\Theta^{[2]}}$.
Note that \eqref{eq.21} implies
\[
 S_t\big\vert \left(\lambda_\kappa^{[1]},\lambda_\kappa^{[2]},\theta^{[1]},\theta^{[2]}, N_{t}\right) \iid {\rm Pois}\left( \lambda_\kappa^{[2]}\theta^{[2]}N_{t}\right).
\]
This model, without the mixture structure and the explanatory variables, is known as Neyman Type A distribution in the collective risk model literature.

\subsection{Premium under frequency-only model with $\mathcal{F}_{T}^{\rm [freq]}$}
We start with the prediction under the classical frequency model discussed in section \ref{sec.freq-only.BMS}. This model is a special case of the frequency--severity model with the severity component ignored.
The Bayesian predictor of the frequency, of Bayesian frequency-premium, based on the past claim history of frequency $\mathcal{F}_{T}^{\rm [freq]}$ can be obtained as
\begin{equation*}
\begin{aligned}
\E{N_{T+1}\big\vert \lambda_\kappa^{[1]}, \lambda_\kappa^{[2]}, \mathcal{F}_{T}^{\rm [freq]}}
&=\E{\E{N_{T+1}\big\vert \lambda_\kappa^{[1]},\lambda_\kappa^{[2]},\Theta^{[1]},\Theta^{[2]}, \mathcal{F}_{T}^{\rm [freq]}}\Big\vert \lambda_\kappa^{[1]}, \lambda_\kappa^{[2]}, \mathcal{F}_{T}^{\rm [freq]} }\\
&=\lambda_\kappa^{[1]}\E{\Theta^{[1]}\big\vert \lambda_\kappa^{[1]}, \lambda_\kappa^{[2]}, \mathcal{F}_{T}^{\rm [freq]}},
\end{aligned}
\end{equation*}
which can take a more specific form as
\begin{equation}\label{eq.b0}
\E{N_{T+1}\big\vert \lambda_\kappa^{[1]}, \lambda_\kappa^{[2]}, \mathcal{F}_{T}^{\rm [freq]}}
=\lambda_\kappa^{\rm [1]}\left[\frac{a_1}{a_1+a_2}\frac{1+\sum_{t=1}^{T} n_t}{c_1} +
\frac{a_2}{a_1+a_2}\frac{1+\sum_{t=1}^{T} n_t}{c_2}\right],
\end{equation}
where
\[
a_1=\frac{\Gamma\left(1+\sum_{t=1}^{T} n_t \right)c_0}{c_1^{\sum_{t=1}^{T} n_t}}\quad\hbox{and}\quad
a_2=\frac{\Gamma\left(1+\sum_{t=1}^{T} n_t \right)(1-c_0)}{c_2^{\sum_{t=1}^{T} n_t}}.
\]
We see that all quantities are functions of the past claim frequencies $n_1,...,n_T$.

\subsection{Premium under frequency--severity model with $\mathcal{F}_{T}^{\rm [freq]}$ }
Let us turn to the Bayesian premium under the frequency--severity model introduced in section \ref{sec.freq-sev.BMS}, with only the past history of claim frequency $\mathcal{F}_{T}^{\rm [freq]}$ incorporated.
The premium in this model is the aggregate severity of the next year. We consider two different scenarios. First, we consider the case where frequency and severity are independent, corresponding to $c_0=0$ or $1$. In this case, the Bayesian premium can be obtained as follows:
\begin{align}
\nonumber & \E{S_{T+1}\big\vert \lambda_\kappa^{[1]}, \lambda_\kappa^{[2]}, \mathcal{F}_{T}^{\rm [freq]}} \\
\nonumber & = \E{N_{T+1}\E{M_{T+1}\big\vert \lambda_\kappa^{[1]}, \lambda_\kappa^{[2]}, \Theta^{[1]}, \Theta^{[2]}, \mathcal{F}_{T}^{\rm [freq]}, N_{T+1} }\Big\vert \lambda_\kappa^{[1]}, \lambda_\kappa^{[2]}, \mathcal{F}_{T}^{\rm [freq]} }\\
\nonumber &= \E{ \lambda_\kappa^{\rm [1]}\lambda_\kappa^{\rm [2]} \Theta^{[1]} \Theta^{[2]}\Big\vert \lambda_\kappa^{[1]}, \lambda_\kappa^{[2]}, \mathcal{F}_{T}^{\rm [freq]}}\\
\nonumber &= \lambda_\kappa^{\rm [1]}\lambda_\kappa^{\rm [2]} \E{ \Theta^{[1]} \Big\vert \lambda_\kappa^{[1]}, \lambda_\kappa^{[2]}, \mathcal{F}_{T}^{\rm [freq]}} \E{\Theta^{[2]}}\\
\label{eq.b1} & =\lambda_\kappa^{\rm [1]}\lambda_\kappa^{\rm [2]}\left[\frac{a_1}{a_1+a_2}\frac{1+\sum_{t=1}^{T} n_t}{c_1} +\frac{a_2}{a_1+a_2}\frac{1+\sum_{t=1}^{T} n_t}{c_2}\right],
\end{align}
where the first and the second equalities are from the law of total expectation and the third equality holds from the independence. The last equality uses $\E{\Theta^{[2]}}=1$.
At this point the analogy between the Bayesian premium and the BMS premium is most relevant and revealing. In this independence case, the Bayesian premium \eqref{eq.b1} is analogous to the BMS premium $\lambda_{\kappa}^{[1]} \lambda_{\kappa}^{[2]} \tilde{r}^{\rm [freq]}(\ell)$ in \eqref{eq.h.9}. Both premiums take the same multiplicative form, and $\tilde{r}^{\rm [freq]}(\ell)$ plays the role of the last term of \eqref{eq.b1}. In particular, we see that the Bayesian premium \eqref{eq.b1} is obtained by just multiplying $\Lambda^{[2]}$ to the Bayesian predictor of frequency \eqref{eq.b0}, which is completely parallel to the BMS premium obtained by multiplying $\lambda_{\kappa}^{[2]}$ to the frequency-only BMS premium $\lambda_\kappa^{[1]}r^{\rm [freq]}(\ell)$ in \eqref{eq.b3}. However, setting the aggregate severity premiums by simply multiplying the severity mean to the frequency-driven premium is invalid unless the independence assumption between frequency and severity is warranted. \\

In the second scenario, we consider the case where the frequency and severity are dependent, corresponding to $0<c_0<1$ in the model.
After straightforward, but tedious, algebra, the Bayesian premium based on $\mathcal{F}_{T}^{\rm [freq]}$ then can be computed as follows:
\begin{align}
\nonumber \E{S_{T+1}\big\vert \lambda_\kappa^{[1]}, \lambda_\kappa^{[2]}, \mathcal{F}_{T}^{\rm [freq]}}
&=\E{\E{S_{T+1}\big\vert \mathcal{F}_{T}^{\rm [freq]}, \lambda_\kappa^{[1]}, \lambda_\kappa^{[2]}, \Theta^{[1]}, \Theta^{[2]}}\Big\vert \lambda_\kappa^{[1]}, \lambda_\kappa^{[2]}, \mathcal{F}_{T}^{\rm [freq]}}\\
\nonumber &=\lambda_\kappa^{[1]}\lambda_\kappa^{[2]}\E{\Theta^{[1]}\Theta^{[2]}\big\vert \lambda_\kappa^{[1]}, \lambda_\kappa^{[2]}, \mathcal{F}_{T}^{\rm [freq]}}\\
\label{eq.b2} &=\lambda_\kappa^{[1]}\lambda_\kappa^{[2]}\left[\frac{a_1^*}{a_1^*+a_2^*}\frac{1+\sum_{t=1}^{T} n_t}{c_1} +
\frac{a_2^*}{a_1^*+a_2^*}\frac{1+\sum_{t=1}^{T} n_t}{c_2}\right],
\end{align}
where
\[
a_1^*= \frac{\Gamma\left(1+\sum_{t=1}^{T} n_t \right)c_0}{c_1^{1+\sum_{t=1}^{T} n_t}}\quad\hbox{and}\quad
a_2^*= \frac{\Gamma\left(1+\sum_{t=1}^{T} n_t \right)(1-c_0)}{c_2^{1+\sum_{t=1}^{T} n_t}}.
\]

Note that, for $0<c_0<1$, we have $a_1^*<a_2^*$, which implies that the premium in \eqref{eq.b2} is larger than the premium in \eqref{eq.b1}. Hence, even though only the frequency history $\mathcal{F}_{T}^{\rm [freq]}$ is used, the Bayesian premium responds to the severity information through the dependence structure between the frequency $\Theta^{[1]}$ and severity $\Theta^{[2]}$, as previously mentioned. Similar to the independence case, we again verify the analogy between the Bayesian premium in \eqref{eq.b2} and the BMS premium of \citet{PengAhn} in \eqref{eq.h.10}.
More specifically, the optimal relativity $\tilde{r}^{\rm [dep]}({l})$ in (\ref{eq.160}) is analogous to the last term of (\ref{eq.b2}). We conclude that, even though no severity history is incorporated, both the Bayesian and BMS premiums can accommodate the severity component indirectly through the dependence structure between frequency and severity.

\subsection{Premium under frequency--severity model with $\mathcal{F}_{T}^{\rm [full]}$ }
\label{sec.4.3}
The most complete premium is the Bayesian premium incorporating both frequency and severity. For this, we use the same frequency--severity model as in the previous subsection. However, this time, we use the full history of the given policyholder, including both frequency and severity, denoted by $\mathcal{F}_{T}^{\rm [full]}$. The corresponding Bayesian premium is then
\begin{equation}\label{eq.22}
\begin{aligned}
\E{S_{T+1}\big\vert \lambda_\kappa^{[1]}, \lambda_\kappa^{[2]}, \mathcal{F}_{T}^{\rm [full]}}
&=\E{\E{S_{T+1}\big\vert \lambda_\kappa^{[1]}, \lambda_\kappa^{[2]}, \Theta^{[1]}, \Theta^{[2]}}, \mathcal{F}_{T}^{\rm [full]}\Big\vert \lambda_\kappa^{[1]}, \lambda_\kappa^{[2]}, \mathcal{F}_{T}^{\rm [full]}}\\
&=\lambda_\kappa^{[1]}\lambda_\kappa^{[2]}\E{\Theta^{[1]}\Theta^{[2]}\big\vert \lambda_\kappa^{[1]}, \lambda_\kappa^{[2]}, \mathcal{F}_{T}^{\rm [full]}}.
\end{aligned}
\end{equation}
After tedious algebraic calculations (shown in Appendix \ref{app.b}), the Bayesian premium can be written as
\begin{equation}\label{eq.26}
\begin{aligned}
&\E{S_{T+1}\big\vert \lambda_\kappa^{[1]}, \lambda_\kappa^{[2]}, \mathcal{F}_{T}^{\rm [full]}}\\
&\quad\quad\quad\quad\quad\quad=\lambda_\kappa^{[1]}\lambda_\kappa^{[2]}
\left[
a_1^{**}\frac{\sum_{t=1}^{T} n_t +1}{\lambda_\kappa^{[1]} T +c_1}
\frac{\sum_{t=1}^{T} s_t +1}{\lambda_\kappa^{[2]} \left(\sum_{t=1}^{T}n_t \right) +c_1}
+a_2^{**}\frac{\sum_{t=1}^{T} n_t +1}{\lambda_\kappa^{[1]} T +c_2}
\frac{\sum_{t=1}^{T} s_t +1}{\lambda_\kappa^{[2]} \left(\sum_{t=1}^{T}n_t \right) +c_2}
\right]
\end{aligned}
\end{equation}
with
\[
a_1^{**}=\frac{
c_0c_1^2
}
{
c_0c_1^2
+
(1-c_0)c_2^2\left(\frac{\lambda_\kappa^{[1]} T + c_1}{\lambda_\kappa^{[1]} T + c_2} \right)^{\sum_{t=1}^{T} n_t +1}
\left(\frac{\lambda_\kappa^{[2]} \left( {\sum_{t=1}^{T} n_t} \right) +c_1}
{\lambda_\kappa^{[2]} \left( {\sum_{t=1}^{T} n_t} \right) +c_2} \right)^{\sum_{t=1}^{T} s_t +1}
}
\]
and $a_2^{**}=1-a_1^{**}$. Clearly, the premium is an explicit function of past severity $s_t$ as well as past frequency $n_t$. Even when frequency and individual severities are independent, that is $c_0=0$ or $1$ in \eqref{eq.a.1}, we obtain
\begin{equation}\label{eq.27}
\E{S_{T+1}\big\vert \lambda_\kappa^{[1]}, \lambda_\kappa^{[2]}, \mathcal{F}_{T}^{\rm [full]}}=
\begin{cases}
\lambda_\kappa^{[1]}\lambda_\kappa^{[2]}\frac{\sum_{t=1}^{T} n_t +1}{\lambda_\kappa^{[1]} T +c_1}
\frac{\sum_{t=1}^{T} s_t +1}{\lambda_\kappa^{[2]} \left(\sum_{t=1}^{T}n_t \right) +c_1}, &c_0=1,\\
\lambda_\kappa^{[1]}\lambda_\kappa^{[2]}\frac{\sum_{t=1}^{T} n_t +1}{\lambda_\kappa^{[1]} T +c_2}
\frac{\sum_{t=1}^{T} s_t +1}{\lambda_\kappa^{[2]} \left(\sum_{t=1}^{T}n_t \right) +c_2}, &c_0=0,\\
\end{cases}
\end{equation}
which still is a function of the history of both frequency and severity.

\subsection{Implication to the BMS}
In view of the analogy between the Bayesian and BMS premium, we require a set of new relativities $\tilde{r}^{\rm [new]}(\ell)$ that incorporates both $n_t$ and $s_t$, or equivalently $\mathcal{F}_{T}^{\rm [full]}$, such that the BMS premium of the form
\begin{equation}
\label{new.BMS.1}
\lambda_{\kappa}^{[1]}\lambda_{\kappa}^{[2]} \tilde{r}^{\rm [new]}(\ell)
\end{equation}
can get as close as possible to the Bayesian premiums \eqref{eq.26} and \eqref{eq.27} with respect to the expected squared error. Unfortunately,
the BMS premium $\lambda_{\kappa}^{[1]}\lambda_{\kappa}^{[2]} \tilde{r}^{\rm [dep]}(\ell)$ proposed by \citet{PengAhn} cannot serve this purpose, since $\tilde{r}^{\rm [dep]}(\ell)$ is determined through $\mathcal{F}_{T}^{\rm [freq]}$ and not $\mathcal{F}_{T}^{\rm [full]}$, as explained in section \ref{sec.freq-sev.BMS}. We will take up the task of finding the new optimal relativity in section \ref{sec.44}. For now, we point out that the new BMS premium (\ref{new.BMS.1}) should reflect the full history $\mathcal{F}_{T}^{\rm [full]}$, as opposed to $\mathcal{F}_{T}^{\rm [freq]}$, and, accordingly, the past severity records should be incorporated in the transition rule in such a BMS structure regardless of whether frequency and severity are independent. The inclusion of severity is important in terms of the efficiency of the prediction as the following formal result capitalizes.

\begin{theorem}\label{thm.1}
 For the frequency--severity model in section \ref{sec.freq-sev.BMS}, we obtain the following inequality of MSEs:
 \[
 \begin{aligned}
 &\E{\left(S_{T+1}-\E{S_{T+1}\big\vert \Lambda^{[1]}, \Lambda^{[2]}, \mathcal{F}_{T}^{\rm [full]}}
 \right)^2
 }\\
 &\quad\quad\quad\quad\quad\quad\le\E{\left(S_{T+1}-\E{S_{T+1}\big\vert \Lambda^{[1]}, \Lambda^{[2]}, \mathcal{F}_{T}^{\rm [freq]}}
 \right)^2}.
 \end{aligned}
 \]
\end{theorem}
\noindent \textbf{Proof:} See Appendix \ref{app.a}. $\qquad \Box$\\

In plain terms, Theorem \ref{thm.1} states that, for the frequency--severity model, adding the past claim history of severity on top of the frequency history always increases the prediction power by providing a smaller MSE. This is true even when the frequency and severity are assumed to be independent, as shown in \eqref{eq.27}. Relatedly, we find conditions where adding the severity history does not improve the prediction performance and whose solution is presented below. The proof is straightforward, and is omitted.
\begin{corollary}\label{corrr.1} If $\P{\Theta^{[2]}=1}=1$, the equality in Theorem \ref{thm.1} holds. That is,
 \begin{equation}\label{eq.a.2}
\E{S_{T+1}\big\vert \lambda_\kappa^{[1]}, \lambda_\kappa^{[2]}, \mathcal{F}_{T}^{\rm [full]}}
=
\E{S_{T+1}\big\vert \lambda_\kappa^{[1]}, \lambda_\kappa^{[2]}, \mathcal{F}_{T}^{\rm [freq]}}.
\end{equation}
\end{corollary}

The condition $\P{\Theta^{[2]}=1}=1$ is equivalent to assuming
only dependence among frequencies is present, removing all other dependencies. Hence, in light of the analogy between the Bayesian premium and BMS premium, the main implication of Theorem \ref{thm.1} is that the inclusion the past claim history of severity in the BMS transition rule can result in a more efficient estimation of the premium.\\

In the next section, we show how to modify the classical BMS by incorporating the transition rule for severity, and, thus, derive the related optimal relativity.

\section{Transition rule driven by frequency--severity history }\label{sec.44}
The optimal relativity proposed by \citet{PengAhn}, as in \eqref{eq.160}, is ``optimal'' in the sense that the BMS transition rule is driven solely by the frequency. It also accommodates the severity component, but only through the dependence structure between frequency and severity and not from the past severity records. To achieve a full optimality, using the historical severity information as well as the frequency in the transition rule is essential, as discussed in section \ref{sec.4}. \\

Considering the severity history in a posteriori risk classification can be traced to \citet{Picard1976}. Since then, various posteriori risk classification methods incorporating both frequency and severity history have been proposed under various collective risk models \citep{Pinquet, Frangos2001, hernandez2009net, Lu2016, Gomez}.
For example, \citet{Picard1976} and \citet{Gomez} consider a posteriori risk classification for frequency by distinguishing small and large claims with some threshold. On the other hand, \citet{Frangos2001}, \citet{hernandez2009net}, and \citet{Lu2016} derive the Bayesian premium for aggregate severity. However, these studies mainly concern distributional quantities, and, thus, are not directly translated to the BMS with proper level structure and the associated transition rules.\\

Though the inclusion of the severity history is important in a posterior risk classification, this element has been largely overlooked in the literature for several reasons. First, historically, the main focus of the BMS in the literature has been the prediction of frequency following the industry tradition. In fact, a typical commercial BMS involves only the claim frequency in most jurisdictions. Second, the frequency-only BMS is simpler to model and analyze. Third, it is implicitly assumed that the severity is independent of the frequency. Hence, the separate analysis of severity has been considered unimportant or redundant. An exception is the recent work of \citet{PengAhn}, where the authors propose a new BMS to predict the aggregate severity under some dependence between frequency and severity. However, their BMS still uses the frequency history only while ignoring the severity history. To this extent, this section proposes a more general BMS and the associated transition rule that distinguishes the degree of claim severity, on top of the classical frequency-based BMS transition rule. We thus derive the optimal relativities.\\

\subsection{New transition rule: $-1/+h_1/+h_2$ system}
To distinguishes the degree of claim severity, we classify each claim into two types based on the claim size following \citet{Picard1976} and \citet{Gomez}. The severity with $y\leq \varphi$ is called a Type I claim and the severity with $y > \varphi$ is called a Type II claim for some predetermined threshold $\varphi$. For a given policyholder, define $L_{t-1}\in \{0, \cdots, z\}$ to be the (random) BM level occupied during time $(t-1, t)$ and assume that the new policyholder starts at level $l_0$ at time 0. A claim-free year in this BMS is rewarded by coming down one BM level; each Type I claim will be penalized by climbing up by $h_1$ levels per claim and each type II claim by $h_2$ levels per claim. We denote this BMS system as the $-1/+h_1/+h_2$ system with threshold $\varphi$. Naturally, we assume $h_2\ge h_1$; if $h=h_1=h_2$, this reduces to the classical $-1/+h$ BMS system \citep{Denuit2}. We denote the number of Type I claims during time $(t-1, t)$ as $N_{t-1}^{\{\rm I\}}$ and the number of Type II claims as $N_{t-1}^{\{\rm II\}}$.\\

Now, consider a driver in Level $L_{t-1} \in \{0,\dots, z\} $ during time $(t-1, t)$. The new level in the next time period $L_t$ can be determined based on the claim history $( N_{t-1}^{\{\rm I\}}, N_{t-1}^{\{\rm II\}} )$ via
\begin{equation}\label{ahn.def.1}
L_t=
\begin{cases}
 \max\left\{L_{t-1}-1, 0 \right\}, &N_{t-1}^{\{\rm I\}}=N_{t-1}^{\{\rm II\}}=0;\\
 \min\left\{L_{t-1}+h_1N_{t-1}^{\{\rm I\}}+h_2N_{t-1}^{\{\rm II\}}, z \right\}, &\hbox{otherwise}.\\
\end{cases}
\end{equation}
Because of the recursive nature, $L_t$ is a function of
$$N_0^{\{I\}}, \cdots, N_{t-1}^{\{I\}}\quad\hbox{and}\quad
N_0^{\{II\}}, \cdots, N_{t-1}^{\{II\}}
$$
for a given $L_0=\ell_0$. Also, as $L_t$ does not depend on $L_0, \cdots, L_{t-2}$, provided that $L_{t-1}$ is given, the process $L_t$ for $t=0, 1, \cdots,$ is Markovian. The transition probabilities of this Markov chain from $L_{t-1}=\ell$ to $L_t=\ell^*$,
given the observed and residual effect characteristics $\lambda_\kappa^{[1]},\lambda_\kappa^{[2]},\theta^{[1]},\theta^{[2]}$,
is defined as
\[
p_{\ell, \ell^*}\left(\lambda_\kappa^{[1]},\lambda_\kappa^{[2]},\theta^{[1]},\theta^{[2]}\right)=\P{L_{t}=\ell^*\big\vert \lambda_\kappa^{[1]},\lambda_\kappa^{[2]},\theta^{[1]},\theta^{[2]}, L_{t-1}=\ell},
\]
and can be calculated as follows.

\begin{theorem}\label{ahn.thm.2}
Consider the frequency--severity model in section \ref{sec.freq-sev.BMS} and a policyholder with
\[
\left(\Lambda^{[1]},\Lambda^{[2]},\Theta^{[1]},\Theta^{[2]} \right)=\left(\lambda_\kappa^{[1]},\lambda_\kappa^{[2]},\theta^{[1]},\theta^{[2]} \right)
\]
under the $-1/+h_1/+h_2$ BMS system with threshold $\varphi$.
Then, for $\ell, \ell^*\in\left\{0,1, \cdots, z\right\}$ and $\ell^*\neq z$, the transition probabilities of moving from level $\ell$ to $\ell^*$ are
\begin{equation}\label{eq.b5}
\begin{aligned}
 &{p_{\ell, \ell^*}\left(\lambda_\kappa^{[1]},\lambda_\kappa^{[2]},\theta^{[1]},\theta^{[2]} \right)}\\
 &=
 \begin{cases}
 \sum\limits_{(k_1, k_2)\in\mathcal{J}(\ell, \ell^*)}\left[
 q_1(k_1+k_2)
 {k_1+k_2 \choose k_1} \left(q_2(\varphi)\right)^{k_2}
 \left(1-q_2(\varphi)\right)^{k_1}\right], &\ell<\ell^*\quad\hbox{or}\quad\ell^*=\ell=z;\\
 q_1(0), &\max\{ \ell-1, 0 \}=\ell^*;\\
 0, & \hbox{otherwise},
 \end{cases}
\end{aligned}
\end{equation}
where $k_1$ and $k_2$ are the number of Type I and II claims, respectively, and
\begin{equation}\label{eq.e1}
\mathcal{J}(\ell, \ell^*)=\left\{ (k_1, k_2)\in\mathbb{N}_0\times\mathbb{N}_0 \big\vert \ell^*=\max\{\ell+k_1h_1+k_2h_2, z\}\right\}
\end{equation}
with $q_1(k)$ and $q_2(\varphi)$ being probabilities defined as
\begin{equation}\label{e1}
 q_1(k)=\P{N_t=k\Big\vert \lambda_\kappa^{[1]},\theta^{[1]}}
\end{equation}
and
\begin{equation}\label{e2}
q_2(\varphi)=\P{Y_{t,j}>\varphi\Big\vert \lambda_\kappa^{[2]},\theta^{[2]} }.
\end{equation}
\end{theorem}
\noindent \textbf{Proof:} See Appendix \ref{app.a}. $\qquad \Box$\\

While the above result is precise, the implementation of the transition probability is cumbersome, mainly because of the complicated definition of the index function in \eqref{eq.e1}. An alternative expression of this quantity, which can be easily implemented in computer, is provided in the following lemma.

\begin{lemma}
The transition probability in \eqref{eq.b5} in Theorem \ref{ahn.thm.2} can be written as
\begin{equation}\label{ahn.eq.7}
\begin{aligned}
 &p_{\ell, \ell^*}\left(\lambda_\kappa^{[1]},\lambda_\kappa^{[2]},\theta^{[1]},\theta^{[2]} \right)\\
 = &\begin{cases}
 \sum\limits_{k_2 \in \mathcal{D}}  \left[
 q_1\left(k_2+\zeta(k_2)\right) \, {\zeta(k_2)+k_2 \choose \zeta(k_2)} \left(q_2(\varphi)\right)^{k_2}
 \left(1-q_2(\varphi) \right)^{\zeta(k_2)} \right], &\ell<\ell^* ;\\
 q_1\left(0\right), & \max{\{\ell -1, 0\}}=\ell^*;\\
 0, &\hbox{otherwise;}
 \end{cases}
\end{aligned}
\end{equation}
for $\ell^*\neq z$, where $q_1$ and $q_2$ are defined in \eqref{e1} and \eqref{e2}, and
\[
\zeta\left(k_2\right)=\frac{(\ell_2-\ell_1) - k_2 h_2 }{h_1}
\quad\hbox{and}\quad
\mathcal{D}=\left\{k_2\in\mathbb{N}_0\big\vert \frac{(\ell_2-\ell_1) - k_2 h_2}{h_1} \in \mathbb{N}_0\right\}.
\]
For $\ell^*=z$, we have
\[
 p_{\ell, z}\left(\lambda_\kappa^{[1]},\lambda_\kappa^{[2]},\theta^{[1]},\theta^{[2]}\right)=1-\sum\limits_{\ell^*=0}^{z-1} p_{\ell, \ell^*}\left(\lambda_\kappa^{[1]},\lambda_\kappa^{[2]},\theta^{[1]},\theta^{[2]}\right).
\]
\end{lemma}
\noindent \textbf{Proof:} It is satisfactory to show \eqref{ahn.eq.7}. For $\ell^*\neq z$, we obtain
 \[
 \ell^*=\max\{\ell+k_1h_1+k_2h_2, z\}
 \]
 if and only if $\ell^*=\ell+k_1h_1+k_2h_2$.
 Hence, we obtain $(k_1, k_2)\in \mathcal{J}(\ell, \ell^*)$ if and only if $k_2\in \mathcal{D}$, which concludes \eqref{ahn.eq.7} from \eqref{eq.b5} for $\ell^*\neq z$.
 $\qquad \Box$\\

\subsection{Long Term Behavior}
From the transition probability ${p_{\ell, \ell^*}\left(\lambda_\kappa^{[1]},\lambda_\kappa^{[2]},\theta^{[1]},\theta^{[2]} \right)}$ in Theorem \ref{ahn.thm.2}, we can construct the one-step transition matrix for the policyholder
\begin{equation}\label{eq.oh.L2}
\boldsymbol{P}\left( \lambda_\kappa^{[1]},\lambda_\kappa^{[2]},\theta^{[1]},\theta^{[2]}\right)
\end{equation}
 of size $(z+1)\times(z+1)$.
The associated conditional stationary distribution is then given by
\begin{equation}\label{eq.oh.L}
\begin{aligned}
&\boldsymbol{\pi}\left(\lambda_\kappa^{[1]},\lambda_\kappa^{[2]}, \theta^{[1]}, \theta^{[2]}\right)\\
&\quad\quad\quad\quad=\left(
\pi_0\left(\lambda_\kappa^{[1]},\lambda_\kappa^{[2]}, \theta^{[1]}, \theta^{[2]}\right), \cdots,
\pi_{z}\left(\lambda_\kappa^{[1]},\lambda_\kappa^{[2]}, \theta^{[1]}, \theta^{[2]}\right)
 \right)^{\mathrm T}
\end{aligned}
\end{equation}
whose $(l+1)$th element is the stationary probability for the policyholder to remain in BM level $\ell$. An explicit expression for the stationary probability from the transition probability is presented in Lemma \ref{lem.2} in Appendix \ref{app.c}.\\

Finally, with this proposed transition rule, the BM level $L$ occupied by the randomly selected policyholder in the steady state, that is, the unconditional stationary distribution, is given as
\begin{equation} \label{eq.PL}
\begin{aligned}
\P{L=\ell}
&=\sum\limits_{\kappa \in\mathcal{K}} w_{\kappa} \iint \pi_\ell\left(\lambda_{\kappa}^{[1]}\theta^{[1]},\lambda_{\kappa}^{[2]}\theta^{[2]} \right) h(\theta^{[1]},\theta^{[2]}) {\rm d}\theta^{[1]}{\rm d}\theta^{[2]}, \quad \hbox{for\quad$\ell=0, \cdots, {z}$}.
\end{aligned}
\end{equation}

\subsection{Optimal Relativities under $-1/+h_1/+h_2$ System}\label{sec.new.relativity}
In this subsection, we develop the optimal relativities for the frequency--severity model in section \ref{sec.freq-sev.BMS} by incorporating the history of severity as well as frequency under the new $-1/+h_1/+h_2$ rule. For a policyholder in the $\kappa$th risk class at the outset of the contract, the BMS premium takes the following form
\[
\lambda_{\kappa}^{[1]}\lambda_{\kappa}^{[2]}{r}^{\rm [new]}(\ell), \qquad \ell=0,...,z,
\]
as previously mentioned in (\ref{new.BMS.1}).
The set of optimal relativities $\tilde{r}^{\rm [new]}(\ell)$ in this setting is the solution of the following optimization problem for the aggregate severity:
 \begin{equation}\label{eq.67}
 \begin{aligned}
&\left(\tilde{r}^{\rm [new]}(0), \cdots, \tilde{r}^{\rm [new]}(z)\right)
=\argmin\limits_{\left(r(0), \cdots, r(z)\right)\in \mathbb{R}^{{z+1}}}
\E{\left(\E{S_{T+1} \big\vert \Lambda^{[1]},\Lambda^{[2]},\Theta^{[1]},\Theta^{[2]}}-\Lambda^{[1]} \Lambda^{[2]} r(L) \right)^2}\\
&\quad\quad\quad\quad\quad\quad\quad\quad\quad=\argmin_{({r}(0), \cdots, {r}(z))\in \mathbb{R}^{{z+1}}}
\sum\limits_{\ell=0}^{z}\E{\left(
\Lambda^{[1]}\Lambda^{[2]}\Theta^{[1]}\Theta^{[2]} -\Lambda^{[1]}\Lambda^{[2]}{r}(L)\right)^2|L=\ell}\P{L=\ell}.\\
\end{aligned}
 \end{equation}
We note that the optimization problem in \eqref{eq.67} is identical to that in \eqref{eq.h.11}, except that the stationary distribution of $L$ in \eqref{eq.67} is now determined by both frequency and severity history as shown in the previous subsection. This is in contrast to \eqref{eq.160}, where the distribution of $L$ is determined by the parameters of frequency history only. Hence, the analytical solution of \eqref{eq.67} can be obtained in a similar manner as \eqref{eq.h.11}.

\begin{theorem}\label{thm.oh.1}
Consider the frequency--severity model in section \ref{sec.freq-sev.BMS} under $-1/+h_1/+h_2$ BMS system.
The optimal relativity $\tilde{r}^{\rm [new]}(\ell)$ that solves \eqref{eq.67} is given by
 \begin{equation}\label{eq.16}
\tilde{r}^{\rm [new]}(\ell)=\frac{\E{\left( \Lambda^{[1]}\Lambda^{[2]}\right)^2 \Theta^{[1]}\Theta^{[2]}\bigg\vert L=\ell}}{\E{\left(\Lambda^{[1]} \Lambda^{[2]}\right)^2\big\vert L=\ell}}
\quad\hbox{for}\quad \ell=0, \cdots,{z},
\end{equation}
where
\begin{equation}\label{eq.161}
\begin{aligned}
&\E{\left( \Lambda^{[1]}\Lambda^{[2]}\right)^2 \Theta^{[1]}\Theta^{[2]}\bigg\vert L=\ell}\\
&=\frac{\sum\limits_{\kappa\in\mathcal{K}} w_{\kappa} \left(\lambda_{\kappa}^{[1]}\lambda_{\kappa}^{[2]}\right)^2
\iint \theta^{[1]}\theta^{[2]} \pi_\ell\left(\lambda_{\kappa}^{[1]}\theta^{[1]}, \lambda_{\kappa}^{[2]}\theta^{[2]}\right) h(\theta^{[1]}, \theta^{[2]}){\rm d}\theta^{[1]}{\rm d} \theta^{[2]}}
{\sum\limits_{\kappa\in\mathcal{K}}w_{\kappa}
\iint \pi_\ell\left(\lambda_{\kappa}^{[1]}\theta^{[1]}, \lambda_{\kappa}^{[2]}\theta^{[2]}
\right) h(\theta^{[1]}, \theta^{[2]}){\rm d}\theta^{[1]}{\rm d} \theta^{[2]} }
\end{aligned}
\end{equation}
and
\begin{equation}\label{eq.162}
\E{\left( \Lambda^{[1]}\Lambda^{[2]}\right)^2\big\vert L=\ell} =
\frac{\sum\limits_{\kappa\in\mathcal{K}} w_{\kappa} \left(\lambda_{\kappa}^{[1]}\lambda_{\kappa}^{[2]}\right)^2 \iint \pi_\ell\left(\lambda_{\kappa}^{[1]}\theta^{[1]}, \lambda_{\kappa}^{[2]}\theta^{[2]}
\right)h(\theta^{[1]}, \theta^{[2]}){\rm d}\theta^{[1]}{\rm d} \theta^{[2]}}
{\sum\limits_{\kappa\in\mathcal{K}} w_{\kappa} \iint \pi_\ell\left(\lambda_{\kappa}^{[1]}\theta^{[1]}, \lambda_{\kappa}^{[2]}\theta^{[2]}
\right) h(\theta^{[1]}, \theta^{[2]}){\rm d}\theta^{[1]}{\rm d} \theta^{[2]}}.
\end{equation}
\end{theorem}
\noindent \textbf{Proof:} See Appendix \ref{app.c}. $\qquad \Box$\\

While the optimal relativity \eqref{eq.16} has a similar expression as in \eqref{eq.160}, it differs in that the distribution of $L$ in \eqref{eq.16} is determined by both frequency and severity components, while the distribution of $L$ in \eqref{eq.160} is determined by the frequency component only. As a result, $\tilde{r}^{\rm [new]}(\ell)$ in \eqref{eq.16} is different from $\tilde{r}^{\rm [dep]}(\ell)$ in \eqref{eq.160}. We remark that the Bayesian version of the BM premium in this section can be represented as follows:
 \begin{equation}\label{kim.2}
\begin{aligned}
\E{S_{T+1}\big\vert \lambda_\kappa^{[1]}, \lambda_\kappa^{[2]}, \mathcal{F}_{T}^{\rm [full]} }
&=\lambda_\kappa^{[1]}\lambda_\kappa^{[2]}\E{\Theta^{[1]}\Theta^{[2]}\big\vert \lambda_\kappa^{[1]}, \lambda_\kappa^{[2]}, \mathcal{F}_{T}^{\rm [full]} }.
\end{aligned}
\end{equation}
Hence, BM relativity $\tilde{r}^{\rm [new]}(\ell)$ can be interpreted as the BMS version of the posterior estimation
\[
\E{\Theta^{[1]}\Theta^{[2]}\big\vert \lambda_\kappa^{[1]}, \lambda_\kappa^{[2]}, \mathcal{F}_{T}^{\rm [full]} }
\]
in \eqref{kim.2}.

\section{Comparison of the BMS}
For the comparison of the quality of the various BMS systems,
\citet{PengAhn} propose the {hypothetical mean square error} (HMSE), a measure defined as the mean square error (MSE) between the aggregate severity and the BMS premium in the stationary state. For example, for the frequency--severity model in section \ref{sec.freq-sev.BMS}, we obtain
\begin{align*}
\label{}
{\rm HMSE}&(\boldsymbol{r}, -1/+h)=\E{\left(\E{S_{T+1} \big\vert \Lambda^{[1]},\Lambda^{[2]},\Theta^{[1]},\Theta^{[2]}}-\Lambda^{[1]} \Lambda^{[2]} r(L) \right)^2}
 \\
 & =\sum\limits_{\kappa\in\mathcal{K}} w_\kappa
\int\int \sum\limits_{l=1}^{z} \left(\lambda_\kappa^{[1]}\lambda_\kappa^{[2]} \theta^{[1]}\theta^{[2]}-
\lambda_\kappa^{[1]}\lambda_\kappa^{[2]}r_l
\right)^2
\pi_\ell\left(\lambda_\kappa^{[1]}\theta^{[1]}\right) h(\theta^{[1]}, \theta^{[2]}) {\rm d}\theta^{[1]}{\rm d}\theta^{[2]}
\end{align*}
for a given relativity set $\boldsymbol{r}=\left(r(0), \cdots, r(z) \right)$. Here $L$ is governed by the transition rule $-1/+h$, and its distribution is the stationary distribution of the BM level of a randomly chosen policyholder. If we use the proposed $-1/+h_1/+h_2$ system (with threshold $\varphi$), the HMSE is defined as
 \begin{equation*}
 \begin{aligned}
 {\rm HMSE}&(\boldsymbol{r}, -1/+h_1/+h_2)\\
&=\E{\left(\E{S_{T+1} \big\vert \Lambda^{[1]},\Lambda^{[2]},\Theta^{[1]},\Theta^{[2]}}-\Lambda^{[1]} \Lambda^{[2]} r_L \right)^2}\\
&=\sum\limits_{\kappa\in\mathcal{K}} w_\kappa
\int\int \sum\limits_{l=1}^{z} \left(\lambda_\kappa^{[1]}\lambda_\kappa^{[2]} \theta^{[1]}\theta^{[2]}-
\lambda_\kappa^{[1]}\lambda_\kappa^{[2]}r_l
\right)^2
\pi_\ell\left(\lambda_\kappa^{[1]}\theta^{[1]}, \lambda_\kappa^{[2]}\theta^{[2]}\right) h(\theta^{[1]}, \theta^{[2]}) {\rm d}\theta^{[1]}{\rm d}\theta^{[2]}.
\end{aligned}
 \end{equation*}
Again, $L$ is governed by the transition rule $-1/+h_1/+h_2$ system with threshold $\varphi$, and its distribution is the stationary distribution of the BM level of a randomly chosen policyholder in this system.
The following corollary justifies the use of the size of severity in the transition rule.

\begin{corollary}\label{HMSE.comp.1}
The premium in BMS with transition rule by the past claim history of both frequency and severity can provide a more efficient premium in the BMS with transition rule by the past claim history of frequency only in the following sense:
\[
 \min\limits_{\boldsymbol{r}\in \mathbb{R}^{{z+1}}, h_1, h_2\in\mathbb{N}_0, \varphi\in\mathbb{R}}{\rm HMSE}(\boldsymbol{r}, -1/+h_1/+h_2)
 \le
 \min\limits_{\boldsymbol{r}\in \mathbb{R}^{{z+1}}, h\in\mathbb{N}_0}
 {\rm HMSE}(\boldsymbol{r}, -1/+h).
\]
\end{corollary}
\noindent \textbf{Proof:}
 The proof is trivial from the observation that the transition rule $-1/+h_1/+h_2$ is the same as the transition rule $-1/+h$ if $h_1=h_2=h$. Consequently, we obtain
 \[
 {\rm HMSE}(\boldsymbol{r}, -1/+h/+h)={\rm HMSE}(\boldsymbol{r}, -1/+h)
 \]
 which concludes the proof.
$\qquad \Box$\\

The remaining section provides a numerical study to show the performance of the proposed BMS compared with the BMS in the recent work of \citet{PengAhn} in section \ref{sec.freq-sev.BMS}. Specifically, we consider the frequency--severity model with the following specification: We assume 10 different BM levels (i.e., $z=9$), but for simplicity, consider a single a priori risk class ($\mathcal{K}=1$).\\

We model the frequency part as
$$N_{t}|(\theta^{[1]},\lambda_\kappa^{[1]})\sim {\rm Poisson}(\lambda_\kappa^{[1]}\theta^{[1]}), $$
and the individual severity part as
 \begin{equation*}
 Y_{t,j}|(\theta^{[2]}, \lambda_\kappa^{[2]}) \sim {\rm Gamma}(\lambda_\kappa^{[2]}\theta^{[2]},1/\psi^{[2]}),
 \end{equation*}
 where $\lambda_\kappa^{[2]}\theta^{[2]}$ is the mean and $1/\psi^{[2]}$ is the shape parameter.
 The bivariate random effect $\left(\Theta^{[1]}, \Theta^{[2]}\right)$ follows the joint distribution
 $H=C(G_1, G_2),$
where $C$ is the Gaussian copula with correlation coefficient $\rho$, and the marginal distributions, $G_1$ and $G_2$, are given by

\[
\begin{cases}
\Theta^{[1]} &\sim {\rm Log Normal}(-\left(\sigma^{[1]}\right)^2/2, \left(\sigma^{[1]}\right)^2);\\
\Theta^{[2]} &\sim {\rm Log Normal}(-\left(\sigma^{[2]}\right)^2/2, \left(\sigma^{[2]}\right)^2),
\end{cases}
\]
and whose parameters are chosen to preserve the unity of the mean. In the following subsections, we investigate the impact of different factors using various selected parameter values.

\subsection{Impact of different dependence between frequency and severity}

We first study the impact of different degrees of dependence between frequency and severity on the HMSE. In particular, we consider three different copula dependence parameters: $\rho=-0.8, -0.4$, and $0.4$. Hence, the larger the value of $\rho$, the stronger the dependence becomes. We assess the impact of these changes on two BMSs: -1/+1 system of \citep{PengAhn} and our proposed -1/+1/+2 system. As our BMS depends on the severity threshold value $\varphi$, we consider four different thresholds, namely, the 75th, 90th, 99th, and 99.9th quantile of severity, for $\varphi$.
The remaining parameters are fixed at $\lambda^{[1]}=0.5$, $\lambda^{[2]}=\exp(8.8)$, $\left(\sigma^{[1]}\right)^2=0.99$, $\left(\sigma^{[2]}\right)^2=0.29$, and $\psi^{[2]}=1/0.67$, which are motivated from the real data analysis in section \ref{sec.6}. Under this specification, the threshold values are given by $\varphi=8200$, $16800$, $48100$, and $94300$. \\

The result is summarized in Table \ref{tab.ex2}, which shows the optimal relativity, stationary distribution, and HMSE for each $\rho=-0.8, -0.4$, and $0.4$. Several comments are in place from this table:
\begin{itemize}
 \item For a given $\rho$ value, as $\varphi$ grows larger, the optimal relativities, stationary distribution, and HMSE become identical for both the -1/+1 and -1/+1/+2 BMSs. This is from a general property that, by definition, the -1/+1/+2 system with infinitely large threshold is the same as the -1/+1 system.
 \item On the other hand, a smaller $\varphi$ means heavier penalty and more drivers end up in higher BM level. This is confirmed by the stationary probability at the highest level $\P{L=9}$ growing larger as $\varphi$ decreases for all $\rho$ values. For $\rho=-0.8$ and $-0.4$, the threshold $\varphi$ at the 90th quantile shows the best performance in terms of the HMSE between four thresholds, whereas, in $\rho=0.4$ case, the 99th quantile is the best threshold. In general, the HMSE of the -1/+1/+2 system varies by the threshold and is not always smaller than that of the -1/+1 system. However, the best HMSE of the -1/+1/+2 system beats that of the -1/+1 system with an appropriate choice of the threshold $\varphi$, following the argument of Corollary \ref{HMSE.comp.1}. The optimal threshold can be found with numerical iterations, though its general solution is not straightforward.

\item As the dependence between frequency and individual severities becomes stronger, the aggregate severity becomes larger and the premium increase will be more substantial. We observe this from comparing the optimal relativities across different $\rho$ values. In particular, as $\rho$ becomes larger, the optimal relativities become larger for most BM levels with the exception in the lower BM levels, especially at level 0, in both the -1/+1 and -1/+1/+2 systems.
\end{itemize}

\subsection{Impact of different mean frequency and transition rule}

To study the impact of different frequency parameter and the transition rule, we consider several combinations on the mean frequency, whose result is shown in Table \ref{tab.ex3}, while other parameters are fixed at $\lambda^{[2]}=\exp(8.8)$, $\psi^{[2]}=1/0.67$, $\left(\sigma^{[1]}\right)^2=0.99$, $\left(\sigma^{[2]}\right)^2=0.29$. For the dependence, we set $\rho=-0.45$, the choice made from an actual dataset in the next section.
\begin{itemize}
 \item In comparing (a) and (b) in Table \ref{tab.ex3}, we note that both keep the same frequency parameter $\lambda^{[1]}$, but Table \ref{tab.ex3}(b) adopts a more stringent transition rule by having the -1/+1/+3 system. Consequently, the stationary probabilities $\P{L=\ell}$ in Table \ref{tab.ex3}(b) tend to be larger in the upper BM levels and smaller in the lower BM levels compared with Table \ref{tab.ex3}(a), as they should.

 \item Furthermore, the best threshold among other alternatives in terms of the HMSE is $\varphi=16800$ in Table \ref{tab.ex3}(a), but $\varphi=48100$ in Table \ref{tab.ex3}(b). This implies that the optimal threshold depends on the transition rule as well as other parameters, as previously mentioned.

 \item Table \ref{tab.ex3}(a) and Table \ref{tab.ex3}(c) share the same transition rule, but have different mean frequency $\lambda^{[1]}$ values. By having a larger frequency mean, Table \ref{tab.ex3}(c) leads to higher proportion of drivers in higher BM levels compared with Table \ref{tab.ex3}(a) under both the -1/+1 and -1/+1/+2 systems.
 However, the relativities are generally smaller in Table \ref{tab.ex3}(c) under both the BMS rules because of the negative $\rho$ value. That is, in this specific setting, the average claim size becomes smaller as there are more claims, and this, in turn, leads to a smaller aggregate loss on average.
 Similar to the above point, we also see that the optimal threshold in terms of the HMSE changes as $\lambda^{[1]}$ does; for $\lambda^{[1]}=0.5$, the optimal threshold is $\varphi=16800$, corresponding to the 90th quantile, whereas $\varphi=94300$, corresponding to the 99.9th quantile, when $\lambda^{[1]}=2$.

\end{itemize}

\subsection{Impact of different dependence among severities}
Here, we study the impact of the dependence among individual severities by changing $\sigma^{[2]}$. We assume the independence between frequency and individual severities by setting $\rho=0$ for this study. The other parameters are fixed at $\lambda^{[1]}=0.5$,
$\lambda^{[2]}=\exp(8.8)$, $\psi^{[2]}=1/0.67$, and $\left(\sigma^{[1]}\right)^2=0.99$. 
\begin{itemize}
 \item Focusing on the -1/+1 system, the optimal relativities remain the same over different $\sigma^{[2]}$ values. This is because the -1/+1 system is driven by the frequency only and its relativities from the independence between frequency and severity; it responds to the mean $E( \Theta^{[2]})$, but this mean is fixed at 1, although $\sigma^{[2]}$ changes. This is easily observable from (\ref{eq.158}).
\item The independence between frequency and severity similarly affects the -1/+1/+2 system, where the stationary distribution stays the same over different $\sigma^{[2]}$ values, which is justified from (\ref{eq.16}). However, the optimal relativities do change. In particular, the relativities become larger in higher BM levels as $\sigma^{[2]}$ increases, for a given threshold $\varphi$. This occurs because a larger $\sigma^{[2]}$ value yields larger claims.
\item As before, the optimal threshold varies as the $\sigma^{[2]}$ value changes. For $\left(\sigma^{[2]}\right)^2=0.29$ and $1$, the optimal threshold in terms of the HMSE is found at the 90th quantile, whereas the optimal threshold is suggested to be much larger for $\left(\sigma^{[2]}\right)^2=0.01$.
\end{itemize}

\section{Data Analysis}\label{sec.6}

\subsection{ Data description and estimation results}
To examine the effect of the dependence on BM relativities, we now consider an actual insurance dataset concerning a collision coverage of new and old vehicles from the Wisconsin Local Government Property Insurance Fund (LGPIF) \citep{Frees4}. Detailed information on the project can be found at the LGPIF project website.
The collision coverage provides cover for the impact of a vehicle with an object, impact of a vehicle with an attached vehicle, and overturn of a vehicle. The observations include policyholders who have either a new collision coverage, an old collision coverage, or both. In our analysis, longitudinal data over the policy years from 2006 to 2010 with 497 governmental entities are used. There are two categorical variables identical to the dataset studied in \citet{PengAhn}: the entity type with six levels and the coverage with three levels, as shown in Table \ref{tab.x}.
Following this, we model the frequency part as
$$N_{it}|(\lambda_{i}^{[1]}, \theta_i^{[1]})\sim {\rm Poisson}(\lambda_i^{[1]}\theta_i^{[1]})\quad \hbox{with} \quad \lambda_i^{[1]}=\exp(\boldsymbol{x}_i^{[1]}\boldsymbol{\beta}^{[1]}), $$
and the individual severity part as
 \begin{equation*}
 Y_{itj}|(\lambda_{i}^{[2]}, \theta_i^{[2]}) \sim {\rm Gamma}(\lambda_i^{[2]}\theta_i^{[2]},1/\psi^{[2]}) \quad \hbox{with} \quad \lambda_i^{[2]}=\exp(\boldsymbol{x}_i^{[2]}\boldsymbol{\beta}^{[2]}),
 \end{equation*}
 where $\lambda_i^{[2]}\theta_i^{[2]}$ is the mean and $1/\psi^{[2]}$ is the shape parameter. Note that we have now introduced an additional subscript $i$ to indicate the $i$th policyholder.
 In this data analysis, we use same a priori risk characteristics for both the frequency and severity parts, that is, $\boldsymbol{x}_i^{[1]}=\boldsymbol{x}_i^{[2]}$,
and $\boldsymbol{\beta}^{[1]}$ and $\boldsymbol{\beta}^{[2]}$ are the corresponding regression coefficient vectors.
For the bivariate random effect $(\Theta^{[1]},\Theta^{[2]})$, we assume bivariate Gaussian copula $C$ in \eqref{eq.4} with the correlation coefficient $\rho$,
and the Lognormal distributions with different parameters for the marginal distributions specified as
\[
\Theta_i^{[1]} \sim {\rm Lognormal}(-\left(\sigma^{[1]}\right)^2/2, \left(\sigma^{[1]}\right)^2) \quad \hbox{and} \quad
\Theta_i^{[2]} \sim {\rm Lognormal}(-\left(\sigma^{[2]}\right)^2/2, \left(\sigma^{[2]}\right)^2).
\]
The estimation results of this model are summarized in Table \ref{est.model2}, and we refer to \citet{PengAhn} for further detailed estimation procedures.

\subsection{Optimal BM relativities in modified BMS}
With the parameter estimates in Table \ref{est.model2}, the optimal BM relativities in \eqref{eq.16} and stationary distribution in \eqref{eq.PL}
are calculated under four different BMS. As per the threshold $\varphi$, we consider the estimated 75th, 90th, 99th, and 99.9th quantile of severity distribution, and the values of the HMSE are also provided to compare the performance. The results are summarized in Table \ref{tab.dat}.\\

In Table \ref{tab.dat}(a) and (b), the BMSs with $h_1=h_2$ are identical to those in Table 2 of \citet{PengAhn}. Paying attention to the -1/+1/+2 and -1/+2/+3 systems, we note that:
\begin{itemize}
\item As $\varphi$ decreases, the proportion of policyholders increases in higher BM levels and decreases in lower BM levels, as seen from the stationary distributions. This occurs because more drivers end up in high BM levels under smaller threshold $\varphi$. This tendency is naturally more pronounced for the -1/+2/+3 system compared with the -1/+1/+2 system, as the penalty under the former system is heavier.

 \item Interestingly, the HMSEs of the optimal relativities under the -1/+1/+2 system are generally smaller than those under the -1/+2/+3 system for each given threshold. Though difficult to explain this in plain terms, we believe that the stationary distributions under the -1/+1/+2 system is more evenly spread than that under the -1/+2/+3 system, which, in turn, marginally stabilizes the BMS and produces smaller relativities. For the relationship between the stationary distribution and the optimal relativity, see Equations (\ref{eq.16})--(\ref{eq.162}).
 \item In terms of the HMSE, we see that the optimal threshold $\varphi$ must be quite large under both the -1/+1/+2 and -1/+2/+3 systems in this particular dataset.
\end{itemize}
Overall, based on the HMSE criterion, we would recommend the -1/+1/+2 transition rule with a highest threshold $\varphi$, the 99.9th quantile, as a new BMS of choice; different datasets may lead to different optimal thresholds. Though the improvement in the HMSE may seem marginal compared with the frequency-driven BMS case in this dataset, the actual amount of dollar saved can be substantial for insurers with large auto portfolios.

\section{Conclusion}
In auto insurance ratemaking, the premium based on a BMS can be understood as a discretely approximated analogy of the Bayesian premium. Traditionally, insurers have been using the frequency-driven BMS, which completely ignores the claim severity information because of mathematical convenience and industry conventions. More sophisticated BMSs have also been proposed based on both frequency and severity of claim, but these adopt the implicit assumption of independence between frequency and severity, which is often invalid in reality. A recent study by \citet{PengAhn} extends these previous works to show how one could include both frequency and severity components in a BMS that targets the aggregate severity when frequency and severity are dependent. However, its transition rule is still determined by the history of claim frequency only. Thus, the resulting BMS premium approximates the Bayesian premium, in a limited sense, and is unable to incorporate the history of claim severity. \\

In the current study, we generalize the approach of \citet{PengAhn} and propose a new BMS. The proposed BMS also targets the aggregate severity, where frequency and severity are dependent, but is more faithful to the purpose of the optimal relativity, that is, estimating of the true premium by reflecting the history of both frequency and severity. Under the proposed BMS, the new transition rule is $-1/-h_1/+h_2$ based on both frequency and severity, with some threshold on the size of severity. With this transition rule, the Bayesian premium can fully incorporate the history of both frequency and severity, and so can the BMS premium from the analogy of the two premiums.

As theoretical contributions, we analytically derive the set of optimal relativities under this new BMS. Furthermore, with the proper choice of the severity threshold and transition rule, the proposed BMS is guaranteed to outperform the alternative frequency-driven BMS in terms of the HMSE. Finding the numerically optimal choices for the bonus and malus scales and the severity threshold is possible by trial and error, but the general solution is seemingly not straightforward to obtain, which could be a possibly interesting research topic in the future.\\

\section*{Acknowledgements}

Rosy Oh was supported by Basic Science Research Program through the
National Research Foundation of Korea(NRF) funded by the Ministry
of Education (Grant No. 2019R1A6A1A11051177). Joseph H.T. Kim is grateful for the support of National Research Foundation of Korea (NRF 2019R1F1A1058095).
Jae Youn Ahn was supported by a National Research Foundation of Korea
(NRF) grant funded by the Korean Government (NRF-2017R1D1A1B03032318).

\bibliographystyle{apalike}
\bibliography{Bib_Yonsei}

\pagebreak

\appendix
\section*{Appendix}
\section{Proofs}\label{app.a}

\begin{proof}[Proof for Theorem \ref{thm.1}] 
Before we provide the proof, we found it
beneficial to define $T_{\ell,\ell^*}(\cdot, \cdot): \{0, 1, \cdots, z\}\times \{0, 1, \cdots, z\} \mapsto\{0, 1\}$, which indicates whether the BM level has changed from $\ell$ in the current year to $\ell^*$ in the next year if there were $k_1$ number of Type I claims and $k_2$ number of Type II claims in the current year. Mathematically, it is defined as, for $\left(k_1, k_2\right) = (0,0)$,
\begin{equation}\label{ahn.eq.1}
 T_{\ell,\ell^*}\left(k_1, k_2\right)=\begin{cases}
 1, & \hbox{if}\quad \max\{\ell-1, 0\}=\ell^*;\\
 0, &\hbox{otherwise};
 \end{cases}
\end{equation}
and, for $\left(k_1, k_2\right)\neq (0,0)$,
\begin{equation}\label{ahn.eq.2}
 T_{\ell,\ell^*}\left(k_1, k_2\right)=\begin{cases}
 1, & \hbox{if}\quad \min\{\ell+h_1 k_1 + h_2 k_2, z\}=\ell^*;\\
 0, &\hbox{otherwise}.
 \end{cases}
\end{equation}

Now, we prove the theorem.
First, observe that
\[
\E{S_{T+1}\big\vert \lambda_\kappa^{[1]}, \lambda_\kappa^{[2]}, \mathcal{F}_T }
=
\argmin\limits_{\eta}\E{\left(\Theta^{[1]}\Theta^{[2]}\lambda_\kappa^{[1]}\lambda_\kappa^{[2]}-\eta\left( \lambda_\kappa^{[1]}, \lambda_\kappa^{[2]}, \mathcal{F}_T\right)
 \right)^2
 \big\vert \lambda_\kappa^{[1]}, \lambda_\kappa^{[2]}, \mathcal{F}_T },
\]
where the right-hand side is optimized among any function $\eta$. Then, the proof follows as such:
\[
\begin{aligned}
&\E{\left(S_{T+1}-\E{S_{T+1}\big\vert \Lambda^{[1]}, \Lambda^{[2]}, \mathcal{F}_{T}^{\rm [full]}}
 \right)^2
 }\\
 &\quad\quad\quad\quad\quad\quad=
 \E{\E{\left(S_{T+1}-\E{S_{T+1}\big\vert \Lambda^{[1]}, \Lambda^{[2]}, \mathcal{F}_{T}^{\rm [full]}}
 \right)^2
 \big\vert \mathcal{F}_T, \Lambda^{[1]}, \Lambda^{[2]}
 }
 }.
 \end{aligned}
\]
\end{proof}

\begin{proof}[Proof for Theorem \ref{ahn.thm.2}] 
 By definition of transition rule in \eqref{ahn.def.1}, it is known that the BM level $L_{t+1}$ is completely explained by
$L_t$ and $\left(N_t^{\{I\}}, N_t^{\{II\}} \right)$.
Now, the transition probability can be calculated as
\begin{equation}\label{eq.b8}
\begin{aligned}
 &p_{\ell, \ell^*}\left(\lambda_\kappa^{[1]},\lambda_\kappa^{[2]},\theta^{[1]},\theta^{[2]} \right)\\
 &=\P{L_{t+1}=\ell^*\big\vert \lambda_\kappa^{[1]},\lambda_\kappa^{[2]},\theta^{[1]},\theta^{[2]}, L_t=\ell }\\
 &=\sum\limits_{k_1=0}^{\infty}\sum\limits_{k_2=0}^{\infty}
 \P{L_{t+1}=\ell^*\big\vert
 \lambda_\kappa^{[1]},\lambda_\kappa^{[2]},\theta^{[1]},\theta^{[2]}, N_t^{\{I\}}=k_1, N_t^{\{II\}}=k_2, L_t=\ell}\\
 &\quad\quad\quad\quad\quad\quad\quad\quad\quad\quad\quad\quad\quad\quad
 \P{N_t^{\{I\}}=k_1, N_t^{\{II\}}=k_2\big\vert
 \lambda_\kappa^{[1]},\lambda_\kappa^{[2]},\theta^{[1]},\theta^{[2]}, L_t=\ell}\\
 &=\sum\limits_{k_1=0}^{\infty}\sum\limits_{k_2=0}^{\infty}
 \P{L_{t+1}=\ell^*\big\vert N_t^{\{I\}}=k_1, N_t^{\{II\}}=k_2, L_t=\ell}\\
 &\quad\quad\quad\quad\quad\quad\quad\quad\quad\quad\quad\quad\quad\quad
 \P{N_t^{\{I\}}=k_1, N_t^{\{II\}}=k_2\big\vert \lambda_\kappa^{[1]},\lambda_\kappa^{[2]},\theta^{[1]},\theta^{[2]}, L_t=\ell}\\
 &=\sum\limits_{k_1=0}^{\infty}\sum\limits_{k_2=0}^{\infty} T_{\ell,\ell^*}(k_1, k_2) \P{N_t^{\{I\}}=k_1, N_t^{\{II\}}=k_2 \big\vert L_t=\ell, \lambda_\kappa^{[1]},\lambda_\kappa^{[2]},\theta^{[1]},\theta^{[2]}}\\
 &=\sum\limits_{k_1 =0}^{\infty}\sum\limits_{k_2 =0}^{\infty} T_{\ell,\ell^*}(k_1, k_2 ) \P{N_t^{\{I\}}=k_1, N_t^{\{II\}}=k_2 \big\vert \lambda_\kappa^{[1]},\lambda_\kappa^{[2]},\theta^{[1]},\theta^{[2]}},\\
\end{aligned}
\end{equation}
where the last equality comes from the independence assumption among frequency, among individual severities, and between frequency and individual severity for given pair of risk characteristics, $\left(\lambda_\kappa^{[1]},\lambda_\kappa^{[2]},\theta^{[1]},\theta^{[2]}\right)$.

Now, we consider the case $\ell<\ell^*$ or $\ell=\ell^*=z$.
Then, for $(k_1, k_2)\in\mathbb{N}_0\times\mathbb{N}_0$, we obtain
\begin{equation}\label{eq.b6}
T_{\ell,\ell^*}(k_1, k_2 )=
\begin{cases}
 1, & \ell^*=\max\{\ell+k_1h1+k_2h_2, z\};\\
 0, & \hbox{elsewhere},
\end{cases}
\end{equation}
and
\begin{equation}\label{eq.b7}
\begin{aligned}
&\P{N_t^{\{I\}}=k_1, N_t^{\{II\}}=k_2 \Big\vert \lambda_\kappa^{[1]},\lambda_\kappa^{[2]},\theta^{[1]},\theta^{[2]}}\\
&=\P{N_t^{\{I\}}=k_1, N_t^{\{II\}}=k_2, N_t=k_1+k_2 \Big\vert \lambda_\kappa^{[1]},\lambda_\kappa^{[2]},\theta^{[1]},\theta^{[2]}}
\\
&=\P{N_t^{\{I\}}=k_1, N_t^{\{II\}}=k_2\Big\vert \lambda_\kappa^{[1]},\lambda_\kappa^{[2]},\theta^{[1]},\theta^{[2]}, N_t=k_1+k_2}
\P{N_t=k_1+k_2 \Big\vert \lambda_\kappa^{[1]},\lambda_\kappa^{[2]},\theta^{[1]},\theta^{[2]}}\\
&=\P{N_t=k_1+k_2\Big\vert \lambda_\kappa^{[1]},\theta^{[1]}}
{k_1+k_2 \choose k_1} \left(\P{Y_{t,j}>\varphi\Big\vert \lambda_\kappa^{[2]},\theta^{[2]} }\right)^{k_2}
 \left(1-\P{Y_{t,j}>\varphi\Big\vert \lambda_\kappa^{[2]},\theta^{[2]} }\right)^{k_1}.
\end{aligned}
\end{equation}
Then, combining \eqref{eq.b8}, \eqref{eq.b6}, and \eqref{eq.b7}, we have the first line in \eqref{eq.b5}.

Now, consider the case $\ell^*=\max\left\{\ell-1, 0 \right\}$. Then, we have
\begin{equation}\label{eq.b10}
T_{\ell,\ell^*}(k_1, k_2 )=
\begin{cases}
 1, & k_1=k_2=0,\\
 0, & \hbox{elsewhere},
\end{cases}
\end{equation}
Hence, combining \eqref{eq.b8} and \eqref{eq.b10}, conclude the second line in \eqref{eq.b5}.

In case the pair $(\ell, \ell^*)$ satisfies none of the following
\[
\ell<\ell^*,\quad \ell=\ell^*=z, \quad\hbox{or}\quad \ell^*=\max\left\{\ell-1, 0 \right\},
\]
we obtain
\[
T_{\ell,\ell^*}(k_1, k_2 )=0
\]
which concludes the proof.
\end{proof}

\begin{proof}[Proof of Theorem \ref{thm.oh.1}]  

Under the frequency--severity model, we obtain
\[
\begin{aligned}
&\E{\left(\E{S_{T+1} \big\vert \Lambda^{[1]},\Lambda^{[2]},\Theta^{[1]},\Theta^{[2]}}-\Lambda^{[1]} \Lambda^{[2]} r^{\rm [new]}(L) \right)^2}\\
&=\sum\limits_{\ell=0}^{z} \E{\left(\E{S_{T+1} \big\vert \Lambda^{[1]},\Lambda^{[2]},\Theta^{[1]},\Theta^{[2]}}-\Lambda^{[1]} \Lambda^{[2]} r^{\rm [new]}(L) \right)^2\Big\vert L=\ell}\P{L=\ell}
\end{aligned}
\]
which further implies that the optimal solution $\tilde{r}^{\rm [new]}(\ell)$ in \eqref{eq.67} satisfies
\begin{equation}\label{eq.p1}
0=\E{-2\Lambda^{[1]} \Lambda^{[2]} \left(\E{S_{T+1} \big\vert \Lambda^{[1]},\Lambda^{[2]},\Theta^{[1]},\Theta^{[2]}} - \tilde{r}^{\rm [new]}(L) \Lambda^{[1]} \Lambda^{[2]} \right)\big\vert L=\ell}
\end{equation}
for $\ell=0, \cdots,{z}$. Then, \eqref{eq.p1} is equivalent with \eqref{eq.16}.

First, we calculate $\E{\left( \Lambda^{[1]}\Lambda^{[2]}\right)^2 \Theta^{[1]}\Theta^{[2]}\bigg\vert L=\ell}$ in \eqref{eq.16} as
\begin{equation*}
\begin{aligned}
&\E{\left( \Lambda^{[1]}\Lambda^{[2]}\right)^2 \Theta^{[1]}\Theta^{[2]}\bigg\vert L=\ell}\\
&=\sum\limits_{\kappa\in\mathcal{K}} \left(\lambda_{\kappa}^{[1]}\lambda_{\kappa}^{[2]}\right)^2
 \E{\Theta^{[1]}\Theta^{[2]} \big\vert \lambda_{\kappa}^{[1]}, \lambda_{\kappa}^{[2]}, L=\ell}\P{\lambda_{\kappa}^{[1]}, \lambda_{\kappa}^{[2]}\big\vert L=\ell}\\
&=\sum\limits_{\kappa\in\mathcal{K}} \left(\lambda_{\kappa}^{[1]}\lambda_{\kappa}^{[2]}\right)^2
 \iint \theta^{[1]}\theta^{[2]} f(\theta^{[1]}, \theta^{[2]}\big\vert \lambda_{\kappa}^{[1]}, \lambda_{\kappa}^{[2]}, L=\ell) {\rm d}\theta^{[1]}{\rm d}\theta^{[2]}\,
 \P{\lambda_{\kappa}^{[1]}, \lambda_{\kappa}^{[2]} \big\vert L=\ell}\\
&=
\frac{\sum\limits_{\kappa\in\mathcal{K}} w_{\kappa} \left(\lambda_{\kappa}^{[1]}\lambda_{\kappa}^{[2]}\right)^2
\iint \theta^{[1]}\theta^{[2]} \pi_\ell\left(\lambda_{\kappa}^{[1]}\theta^{[1]}, \lambda_{\kappa}^{[2]}\theta^{[2]}\right) h(\theta^{[1]}, \theta^{[2]}){\rm d}\theta^{[1]}{\rm d} \theta^{[2]}}
{\sum\limits_{\kappa\in\mathcal{K}}w_{\kappa}
\iint \pi_\ell\left(\lambda_{\kappa}^{[1]}\theta^{[1]}, \lambda_{\kappa}^{[2]}\theta^{[2]}
\right) h(\theta^{[1]}, \theta^{[2]}){\rm d}\theta^{[1]}{\rm d} \theta^{[2]} }
,
\end{aligned}
\end{equation*}
 where the last equation comes from the independence between the a priori and posteriori.
Similarly, we can calculate $\E{\left( \Lambda^{[1]}\Lambda^{[2]}\right)^2\big\vert L=\ell} $ in \eqref{eq.16} as
\[
\begin{aligned}
\E{\left( \Lambda^{[1]}\Lambda^{[2]}\right)^2\big\vert L=\ell}
&=\sum\limits_{\kappa\in\mathcal{K}} \left(\lambda_{\kappa}^{[1]}\lambda_{\kappa}^{[2]}\right)^2
 \P{ \lambda_{\kappa}^{[1]}, \lambda_{\kappa}^{[2]}\big\vert L=\ell}\\
&=\sum\limits_{\kappa\in\mathcal{K}} \left(\lambda_{\kappa}^{[1]}\lambda_{\kappa}^{[2]}\right)^2 \frac{\P{L=\ell\big\vert \lambda_{\kappa}^{[1]}, \lambda_{\kappa}^{[2]}}\P{\lambda_{\kappa}^{[1]}, \lambda_{\kappa}^{[2]}} }{\P{L=\ell}}\\
&=\frac{\sum\limits_{\kappa\in\mathcal{K}} w_{\kappa} \left(\lambda_{\kappa}^{[1]}\lambda_{\kappa}^{[2]}\right)^2 \iint \pi_\ell\left(\lambda_{\kappa}^{[1]}\theta^{[1]}, \lambda_{\kappa}^{[2]}\theta^{[2]}
\right)h(\theta^{[1]}, \theta^{[2]}){\rm d}\theta^{[1]}{\rm d} \theta^{[2]}}
{\sum\limits_{\kappa\in\mathcal{K}} w_{\kappa} \iint \pi_\ell\left(\lambda_{\kappa}^{[1]}\theta^{[1]}, \lambda_{\kappa}^{[2]}\theta^{[2]}
\right) h(\theta^{[1]}, \theta^{[2]}){\rm d}\theta^{[1]}{\rm d} \theta^{[2]}},
\end{aligned}
\]
 where the last equation comes from the independence between the a priori and posteriori.

\end{proof}

\section{Derivation of Bayesian Premium}\label{app.b}

For the derivation of \eqref{eq.26} in section \eqref{sec.4.3} from \eqref{eq.22}, it is satisfactory to derive the following conditional distribution of the bivariate residual effect:
\begin{equation}\label{eq.24}
\begin{aligned}
&h\left(\theta_1, \theta_2\big\vert \lambda_\kappa^{[1]}, \lambda_\kappa^{[2]}, \mathcal{F}_{T}^{\rm [full]}\right) \\
&\quad=
h\left(\theta_1, \theta_2\big\vert (n_1, s_1), \cdots, (n_T, s_T), \lambda_\kappa^{[1]}, \lambda_\kappa^{[2]}\right)\\
&\quad=\frac{h\left(\theta_1, \theta_2\big\vert \lambda_\kappa^{[1]}, \lambda_\kappa^{[2]} \right)\prod\limits_{t=1}^{T}\left[f_1\left(n_t\big\vert \lambda_\kappa^{[1]}, \lambda_\kappa^{[2]}, \theta_1, \theta_2\right)f_2\left(s_t\big\vert \lambda_\kappa^{[1]}, \lambda_\kappa^{[2]}, \theta_1, \theta_2, n_t \right)\right]}{\prod\limits_{t=1}^{T} f_1\left(n_t\big\vert \lambda_\kappa^{[1]}, \lambda_\kappa^{[2]} \right)f_2\left(s_t\big\vert \lambda_\kappa^{[1]}, \lambda_\kappa^{[2]}, n_t\right)}\\
&\quad=\frac{h\left(\theta_1, \theta_2 \right)\prod\limits_{t=1}^{T}\left[f_1\left(n_t\big\vert \lambda_\kappa^{[1]}, \lambda_\kappa^{[2]}, \theta_1, \theta_2 \right)
f_2\left(s_t\big\vert \lambda_\kappa^{[1]}, \lambda_\kappa^{[2]}, \theta_1, \theta_2, n_t \right)\right]}{\prod\limits_{t=1}^{T} f_1\left(n_t\big\vert \lambda_\kappa^{[1]}, \lambda_\kappa^{[2]} \right)f_2\left(s_t\big\vert \lambda_\kappa^{[1]}, \lambda_\kappa^{[2]}, n_t\right)}.\\
\end{aligned}
\end{equation}
After tedious calculation, \eqref{eq.24} can be obtained as follow:
\begin{equation}\label{eq.25}
\begin{aligned}
&h\left(\theta_1, \theta_2\big\vert \lambda_\kappa^{[1]}, \lambda_\kappa^{[2]}, \mathcal{F}_{T}^{\rm [full]}\right)\\
&=
a_1^{**} {\rm dgamma}\left( \theta_1, \sum_{t=1}^{T} n_t +1, \lambda_\kappa^{[1]}T+c_1\right)
{\rm dgamma}\left( \theta_2, \sum_{t=1}^{T} s_t +1, \lambda_\kappa^{[2]} \left(\sum_{t=1}^{T}n_t \right) +c_1\right)\\
&\quad +a_2^{**} {\rm dgamma}\left( \theta_1, \sum_{t=1}^{T} n_t +1, \lambda_\kappa^{[1]}T+c_2\right)
{\rm dgamma}\left( \theta_2, \sum_{t=1}^{T} s_t +1, \lambda_\kappa^{[2]} \left(\sum_{t=1}^{T}n_t \right) +c_2 \right),\\
\end{aligned}
\end{equation}
where ${\rm dgamma}(\cdot, \alpha_1, \alpha_2)$ is a density function of gamma distribution of shape $=\alpha_1$ and rate $=\alpha_2$.
Finally, from \eqref{eq.25}, the premium defined in \eqref{eq.22} can be obtained.


\section{Miscellaneous}\label{app.c}

\begin{lemma}\label{lem.2}
Assume
\begin{equation}\label{eq.11}
\P{N_t>0\big\vert \lambda_\kappa^{[1]},\lambda_\kappa^{[2]},\theta^{[1]},\theta^{[2]}}>0
\end{equation}
and
\begin{equation}\label{eq.12}
\P{N_t=0\big\vert \lambda_\kappa^{[1]},\lambda_\kappa^{[2]},\theta^{[1]},\theta^{[2]}}>0.
\end{equation}
Then, for the given observed and residual effect characteristic $\lambda_\kappa^{[1]},\lambda_\kappa^{[2]}, \theta^{[1]}, \theta^{[2]}$,
we obtain
\[
\begin{aligned}
 \boldsymbol{\pi}\left(\lambda_\kappa^{[1]},\lambda_\kappa^{[2]},\theta^{[1]},\theta^{[2]}\right)
 &=\boldsymbol{e}^{\mathrm T}\left(\boldsymbol{I}- \boldsymbol{P}\left(\lambda_\kappa^{[1]},\lambda_\kappa^{[2]},\theta^{[1]},\theta^{[2]}\right) + \boldsymbol{E}\right)^{-1},
\end{aligned}
\]
where $\boldsymbol{E}$ is the $(z+1)\times (z+1)$ matrix all of whose entries are 1 and $\boldsymbol{e}$ is $(z+1)$ column vector all of whose entries are 1.
\end{lemma}
\begin{proof}
 Note that \eqref{eq.11} and \eqref{eq.12} implies that the matrix
 \[
 \boldsymbol{P}\left(\Lambda^{[1]},\Lambda^{[2]},\Theta^{[1]},\Theta^{[2]}\right)
 \]
 is regular. Hence, the remaining proof follows from Property 4.2 in \citet{Denuit2}, for example, among many other references.
\end{proof}

\begin{table}[p]
 \caption{ Optimal BM relativities and stationary distribution of $L$ for various threshold $\varphi$ and dependence parameter $\rho$ under the two BMSs: the -1/+1 and -1/+1/+2 systems } \label{tab.ex2}
 \centering
 \resizebox{\linewidth}{!}{
 \begin{tabular}{ l c c c c c c c c c c c c c c c c c c c }
 \\
 \multicolumn{12}{l}{(a) For $\rho=-0.8$}\\

 \hline
&& \multicolumn{2}{l}{-1/+1 system} && \multicolumn{6}{l}{-1/+1/+2 system}\\\cline{3-4} \cline{6-16}
 & $\varphi$& \multicolumn{2}{c}{-} && \multicolumn{2}{c}{8200} && \multicolumn{2}{c}{16800} && \multicolumn{2}{c}{48100}&& \multicolumn{2}{c}{94300} \\\cline{3-4} \cline{6-7}\cline{9-10}\cline{12-13}\cline{15-16}
Level $\ell$ && $r(\ell)$ & $\P{L=\ell}$ && $r(\ell)$ & $\P{L=\ell}$ && $r(\ell)$ & $\P{L=\ell}$ && $r(\ell)$ & $\P{L=\ell}$ && $r(\ell)$ & $\P{L=\ell}$ \\
 \hline
9	&&	1.328 	&	0.135 	&&	1.292 	&	0.148 	&&	1.320 	&	0.139 	&&	1.328 	&	0.135 	&&	1.328 	&	0.135 	\\
8	&&	1.052 	&	0.055 	&&	1.018 	&	0.064 	&&	1.046 	&	0.057 	&&	1.052 	&	0.055 	&&	1.052 	&	0.055 	\\
7	&&	0.936 	&	0.034 	&&	0.898 	&	0.041 	&&	0.928 	&	0.036 	&&	0.936 	&	0.034 	&&	0.936 	&	0.034 	\\
6	&&	0.858 	&	0.026 	&&	0.815 	&	0.033 	&&	0.849 	&	0.028 	&&	0.859 	&	0.026 	&&	0.858 	&	0.026 	\\
5	&&	0.795 	&	0.024 	&&	0.746 	&	0.031 	&&	0.782 	&	0.027 	&&	0.796 	&	0.024 	&&	0.796 	&	0.024 	\\
4	&&	0.737 	&	0.026 	&&	0.681 	&	0.034 	&&	0.719 	&	0.030 	&&	0.737 	&	0.027 	&&	0.737 	&	0.026 	\\
3	&&	0.676 	&	0.034 	&&	0.619 	&	0.043 	&&	0.654 	&	0.039 	&&	0.675 	&	0.035 	&&	0.676 	&	0.034 	\\
2	&&	0.607 	&	0.055 	&&	0.545 	&	0.067 	&&	0.579 	&	0.061 	&&	0.604 	&	0.056 	&&	0.606 	&	0.055 	\\
1	&&	0.522 	&	0.114 	&&	0.493 	&	0.096 	&&	0.510 	&	0.107 	&&	0.521 	&	0.113 	&&	0.522 	&	0.114 	\\
0	&&	0.414 	&	0.496 	&&	0.395 	&	0.445 	&&	0.406 	&	0.476 	&&	0.413 	&	0.495 	&&	0.414 	&	0.496 	\\
\hline\hline
HMSE	&& \multicolumn{2}{c}{1.297}	&&	\multicolumn{2}{c}{1.293}&&	\multicolumn{2}{c}{1.282}&&	\multicolumn{2}{c}{1.295}&&	\multicolumn{2}{c}{1.297}	\\ \hline
\end{tabular}
}

 \bigskip

 \resizebox{\linewidth}{!}{
 \begin{tabular}{ l c c c c c c c c c c c c c c c c c c c }
 \multicolumn{12}{l}{(b) For $\rho=-0.4$}\\

 \hline
&& \multicolumn{2}{l}{-1/+1 system} && \multicolumn{6}{l}{-1/+1/+2 system}\\\cline{3-4} \cline{6-16}
 & $\varphi$& \multicolumn{2}{c}{-} && \multicolumn{2}{c}{8200} && \multicolumn{2}{c}{16800} && \multicolumn{2}{c}{48100}&& \multicolumn{2}{c}{94300} \\\cline{3-4} \cline{6-7}\cline{9-10}\cline{12-13}\cline{15-16}
Level $\ell$ && $r(\ell)$ & $\P{L=\ell}$ && $r(\ell)$ & $\P{L=\ell}$ && $r(\ell)$ & $\P{L=\ell}$ && $r(\ell)$ & $\P{L=\ell}$ && $r(\ell)$ & $\P{L=\ell}$ \\
 \hline
9	&&	2.047 	&	0.135 	&&	1.968 	&	0.151 	&&	2.026 	&	0.141 	&&	2.047 	&	0.136 	&&	2.047 	&	0.135 	\\
8	&&	1.460 	&	0.055 	&&	1.398 	&	0.065 	&&	1.449 	&	0.059 	&&	1.462 	&	0.055 	&&	1.461 	&	0.055 	\\
7	&&	1.237 	&	0.034 	&&	1.173 	&	0.041 	&&	1.225 	&	0.037 	&&	1.240 	&	0.034 	&&	1.238 	&	0.034 	\\
6	&&	1.096 	&	0.026 	&&	1.025 	&	0.033 	&&	1.080 	&	0.029 	&&	1.099 	&	0.026 	&&	1.097 	&	0.026 	\\
5	&&	0.987 	&	0.024 	&&	0.907 	&	0.031 	&&	0.964 	&	0.027 	&&	0.989 	&	0.024 	&&	0.987 	&	0.024 	\\
4	&&	0.889 	&	0.026 	&&	0.801 	&	0.033 	&&	0.859 	&	0.029 	&&	0.890 	&	0.027 	&&	0.889 	&	0.026 	\\
3	&&	0.791 	&	0.034 	&&	0.701 	&	0.041 	&&	0.754 	&	0.038 	&&	0.790 	&	0.035 	&&	0.791 	&	0.035 	\\
2	&&	0.683 	&	0.055 	&&	0.595 	&	0.064 	&&	0.643 	&	0.060 	&&	0.680 	&	0.056 	&&	0.683 	&	0.055 	\\
1	&&	0.559 	&	0.114 	&&	0.501 	&	0.095 	&&	0.527 	&	0.106 	&&	0.554 	&	0.113 	&&	0.559 	&	0.114 	\\
0	&&	0.411 	&	0.496 	&&	0.376 	&	0.445 	&&	0.392 	&	0.475 	&&	0.408 	&	0.494 	&&	0.411 	&	0.496 	\\
\hline\hline
HMSE	&& \multicolumn{2}{c}{6.022}&&	\multicolumn{2}{c}{5.995}&&	\multicolumn{2}{c}{5.930}&&	\multicolumn{2}{c}{5.996}&&	\multicolumn{2}{c}{6.018}	\\ \hline
\end{tabular}
}

 \bigskip
 \resizebox{\linewidth}{!}{
 \begin{tabular}{ l c c c c c c c c c c c c c c c c c c c }
 \multicolumn{12}{l}{(c) For $\rho=0.4$}\\

 \hline
&& \multicolumn{2}{l}{-1/+1 system} && \multicolumn{6}{l}{-1/+1/+2 system}\\\cline{3-4} \cline{6-16}
 & $\varphi$& \multicolumn{2}{c}{-} && \multicolumn{2}{c}{8200} && \multicolumn{2}{c}{16800} && \multicolumn{2}{c}{48100}&& \multicolumn{2}{c}{94300} \\\cline{3-4} \cline{6-7}\cline{9-10}\cline{12-13}\cline{15-16}
Level $\ell$ && $r(\ell)$ & $\P{L=\ell}$ && $r(\ell)$ & $\P{L=\ell}$ && $r(\ell)$ & $\P{L=\ell}$ && $r(\ell)$ & $\P{L=\ell}$ && $r(\ell)$ & $\P{L=\ell}$ \\
 \hline
9	&&	4.489 	&	0.135 	&&	4.129 	&	0.156 	&&	4.338 	&	0.145 	&&	4.478 	&	0.136 	&&	4.489 	&	0.135 	\\
8	&&	2.512 	&	0.055 	&&	2.285 	&	0.067 	&&	2.423 	&	0.060 	&&	2.508 	&	0.055 	&&	2.512 	&	0.055 	\\
7	&&	1.918 	&	0.034 	&&	1.715 	&	0.042 	&&	1.837 	&	0.037 	&&	1.913 	&	0.034 	&&	1.918 	&	0.034 	\\
6	&&	1.581 	&	0.026 	&&	1.387 	&	0.033 	&&	1.501 	&	0.029 	&&	1.575 	&	0.026 	&&	1.580 	&	0.026 	\\
5	&&	1.341 	&	0.024 	&&	1.154 	&	0.030 	&&	1.260 	&	0.027 	&&	1.334 	&	0.024 	&&	1.340 	&	0.024 	\\
4	&&	1.142 	&	0.026 	&&	0.964 	&	0.032 	&&	1.062 	&	0.029 	&&	1.134 	&	0.027 	&&	1.141 	&	0.026 	\\
3	&&	0.958 	&	0.034 	&&	0.797 	&	0.039 	&&	0.881 	&	0.037 	&&	0.948 	&	0.035 	&&	0.957 	&	0.034 	\\
2	&&	0.772 	&	0.055 	&&	0.636 	&	0.060 	&&	0.704 	&	0.057 	&&	0.762 	&	0.055 	&&	0.771 	&	0.055 	\\
1	&&	0.578 	&	0.114 	&&	0.487 	&	0.094 	&&	0.524 	&	0.105 	&&	0.567 	&	0.113 	&&	0.576 	&	0.114 	\\
0	&&	0.371 	&	0.496 	&&	0.322 	&	0.448 	&&	0.344 	&	0.476 	&&	0.366 	&	0.494 	&&	0.370 	&	0.496 	\\
\hline\hline
HMSE	&& \multicolumn{2}{c}{53.792}	&&	\multicolumn{2}{c}{55.196}&&	\multicolumn{2}{c}{54.138}&&	\multicolumn{2}{c}{53.725}&&	\multicolumn{2}{c}{53.772}	\\ \hline
\end{tabular}
}

\end{table}


\begin{table}[p]
 \caption{Optimal BM relativities and stationary distribution of $L$ for various threshold $\varphi$ and $\lambda^{[1]}$ under the three BMSs: the -1/+1, -1/+1/+2, and -1/+1/+3 systems } \label{tab.ex3}
 \centering
 \resizebox{\linewidth}{!}{
 \begin{tabular}{ l c c c c c c c c c c c c c c c c c c c }
 \\
 \multicolumn{12}{l}{(a) -1/+1 system and -1/+1/+2 system for $\lambda^{[1]}=0.5$}\\

 \hline
&& \multicolumn{2}{l}{-1/+1 system} && \multicolumn{6}{l}{-1/+1/+2 system}\\\cline{3-4} \cline{6-16}
 & $\varphi$& \multicolumn{2}{c}{-} && \multicolumn{2}{c}{8200} && \multicolumn{2}{c}{16800} && \multicolumn{2}{c}{48100}&& \multicolumn{2}{c}{94300} \\\cline{3-4} \cline{6-7}\cline{9-10}\cline{12-13}\cline{15-16}
Level $\ell$ && $r(\ell)$ & $\P{L=\ell}$ && $r(\ell)$ & $\P{L=\ell}$ && $r(\ell)$ & $\P{L=\ell}$ && $r(\ell)$ & $\P{L=\ell}$ && $r(\ell)$ & $\P{L=\ell}$ \\
 \hline
9	&&	3.070 	&	0.135 	&&	2.893 	&	0.154 	&&	3.009 	&	0.143 	&&	3.069 	&	0.136 	&&	3.071 	&	0.135 	\\
8	&&	1.951 	&	0.055 	&&	1.827 	&	0.066 	&&	1.915 	&	0.059 	&&	1.954 	&	0.055 	&&	1.952 	&	0.055 	\\
7	&&	1.572 	&	0.034 	&&	1.451 	&	0.042 	&&	1.535 	&	0.037 	&&	1.574 	&	0.034 	&&	1.573 	&	0.034 	\\
6	&&	1.344 	&	0.026 	&&	1.220 	&	0.033 	&&	1.303 	&	0.029 	&&	1.345 	&	0.026 	&&	1.345 	&	0.026 	\\
5	&&	1.174 	&	0.024 	&&	1.046 	&	0.030 	&&	1.127 	&	0.027 	&&	1.174 	&	0.024 	&&	1.175 	&	0.024 	\\
4	&&	1.029 	&	0.026 	&&	0.898 	&	0.033 	&&	0.976 	&	0.029 	&&	1.026 	&	0.027 	&&	1.029 	&	0.026 	\\
3	&&	0.888 	&	0.034 	&&	0.761 	&	0.040 	&&	0.831 	&	0.037 	&&	0.883 	&	0.035 	&&	0.888 	&	0.035 	\\
2	&&	0.741 	&	0.055 	&&	0.625 	&	0.062 	&&	0.684 	&	0.058 	&&	0.733 	&	0.056 	&&	0.740 	&	0.055 	\\
1	&&	0.579 	&	0.114 	&&	0.499 	&	0.094 	&&	0.531 	&	0.105 	&&	0.569 	&	0.113 	&&	0.577 	&	0.114 	\\
0	&&	0.396 	&	0.496 	&&	0.352 	&	0.447 	&&	0.371 	&	0.475 	&&	0.392 	&	0.494 	&&	0.396 	&	0.496 	\\
\hline\hline
HMSE	&& \multicolumn{2}{c}{19.120}	&&	\multicolumn{2}{c}{19.345}&& \multicolumn{2}{c}{19.025}&&	\multicolumn{2}{c}{19.049}&&	\multicolumn{2}{c}{19.105}	\\ \hline
\end{tabular}
}

 \bigskip

 \resizebox{\linewidth}{!}{
 \begin{tabular}{ l c c c c c c c c c c c c c c c c c c c }
 \multicolumn{12}{l}{(b) -1/+2 system and -1/+2/+3 system for $\lambda^{[1]}=0.5$}\\

 \hline
&& \multicolumn{2}{l}{-1/+2 system} && \multicolumn{6}{l}{-1/+2/+3 system}\\\cline{3-4} \cline{6-16}
 & $\varphi$& \multicolumn{2}{c}{-} && \multicolumn{2}{c}{8200} && \multicolumn{2}{c}{16800} && \multicolumn{2}{c}{48100}&& \multicolumn{2}{c}{94300} \\\cline{3-4} \cline{6-7}\cline{9-10}\cline{12-13}\cline{15-16}
Level $\ell$ && $r(\ell)$ & $\P{L=\ell}$ && $r(\ell)$ & $\P{L=\ell}$ && $r(\ell)$ & $\P{L=\ell}$ && $r(\ell)$ & $\P{L=\ell}$ && $r(\ell)$ & $\P{L=\ell}$ \\
 \hline
9	&&	2.438 	&	0.202 	&&	2.367 	&	0.213 	&&	2.412 	&	0.206 	&&	2.437 	&	0.202 	&&	2.438 	&	0.202 	\\
8	&&	1.493 	&	0.093 	&&	1.446 	&	0.100 	&&	1.478 	&	0.096 	&&	1.493 	&	0.093 	&&	1.493 	&	0.093 	\\
7	&&	1.161 	&	0.059 	&&	1.115 	&	0.065 	&&	1.145 	&	0.062 	&&	1.160 	&	0.060 	&&	1.161 	&	0.059 	\\
6	&&	0.949 	&	0.047 	&&	0.911 	&	0.050 	&&	0.935 	&	0.048 	&&	0.948 	&	0.047 	&&	0.949 	&	0.047 	\\
5	&&	0.815 	&	0.040 	&&	0.766 	&	0.045 	&&	0.796 	&	0.042 	&&	0.814 	&	0.040 	&&	0.815 	&	0.040 	\\
4	&&	0.676 	&	0.045 	&&	0.643 	&	0.044 	&&	0.659 	&	0.044 	&&	0.673 	&	0.045 	&&	0.675 	&	0.045 	\\
3	&&	0.606 	&	0.042 	&&	0.558 	&	0.048 	&&	0.587 	&	0.045 	&&	0.605 	&	0.043 	&&	0.606 	&	0.042 	\\
2	&&	0.476 	&	0.074 	&&	0.445 	&	0.066 	&&	0.457 	&	0.071 	&&	0.472 	&	0.074 	&&	0.475 	&	0.074 	\\
1	&&	0.446 	&	0.059 	&&	0.417 	&	0.053 	&&	0.429 	&	0.056 	&&	0.443 	&	0.059 	&&	0.446 	&	0.059 	\\
0	&&	0.312 	&	0.338 	&&	0.293 	&	0.316 	&&	0.301 	&	0.329 	&&	0.310 	&	0.337 	&&	0.311 	&	0.338 	\\
\hline\hline
HMSE	&& \multicolumn{2}{c}{21.288}	&& \multicolumn{2}{c}{21.509}	&&	\multicolumn{2}{c}{21.325}&&	\multicolumn{2}{c}{21.275}&&	\multicolumn{2}{c}{21.285}	\\ \hline
\end{tabular}
}

 \bigskip

 \resizebox{\linewidth}{!}{
 \begin{tabular}{ l c c c c c c c c c c c c c c c c c c c }
 \multicolumn{12}{l}{(c) -1/+1 system and -1/+1/+2 system for $\lambda^{[1]}=2$}\\

 \hline
&& \multicolumn{2}{l}{-1/+1 system} && \multicolumn{6}{l}{-1/+1/+2 system}\\\cline{3-4} \cline{6-16}
 & $\varphi$& \multicolumn{2}{c}{-} && \multicolumn{2}{c}{8200} && \multicolumn{2}{c}{16800} && \multicolumn{2}{c}{48100}&& \multicolumn{2}{c}{94300} \\\cline{3-4} \cline{6-7}\cline{9-10}\cline{12-13}\cline{15-16}
Level $\ell$ && $r(\ell)$ & $\P{L=\ell}$ && $r(\ell)$ & $\P{L=\ell}$ && $r(\ell)$ & $\P{L=\ell}$ && $r(\ell)$ & $\P{L=\ell}$ && $r(\ell)$ & $\P{L=\ell}$ \\
 \hline
9	&&	1.515 	&	0.554 	&&	1.479 	&	0.574 	&&	1.501 	&	0.563 	&&	1.514 	&	0.555 	&&	1.515 	&	0.554 	\\
8	&&	0.640 	&	0.111 	&&	0.616 	&	0.119 	&&	0.631 	&	0.115 	&&	0.640 	&	0.112 	&&	0.640 	&	0.111 	\\
7	&&	0.476 	&	0.050 	&&	0.450 	&	0.054 	&&	0.466 	&	0.052 	&&	0.475 	&	0.050 	&&	0.476 	&	0.050 	\\
6	&&	0.393 	&	0.031 	&&	0.364 	&	0.033 	&&	0.380 	&	0.032 	&&	0.392 	&	0.031 	&&	0.393 	&	0.031 	\\
5	&&	0.339 	&	0.024 	&&	0.307 	&	0.025 	&&	0.324 	&	0.024 	&&	0.338 	&	0.024 	&&	0.339 	&	0.024 	\\
4	&&	0.298 	&	0.021 	&&	0.265 	&	0.022 	&&	0.281 	&	0.022 	&&	0.296 	&	0.022 	&&	0.298 	&	0.021 	\\
3	&&	0.263 	&	0.023 	&&	0.228 	&	0.022 	&&	0.245 	&	0.023 	&&	0.260 	&	0.023 	&&	0.262 	&	0.023 	\\
2	&&	0.229 	&	0.029 	&&	0.196 	&	0.026 	&&	0.210 	&	0.028 	&&	0.226 	&	0.029 	&&	0.229 	&	0.029 	\\
1	&&	0.194 	&	0.043 	&&	0.165 	&	0.032 	&&	0.176 	&	0.038 	&&	0.190 	&	0.042 	&&	0.194 	&	0.043 	\\
0	&&	0.154 	&	0.114 	&&	0.133 	&	0.093 	&&	0.141 	&	0.105 	&&	0.152 	&	0.113 	&&	0.154 	&	0.114 	\\
\hline\hline
HMSE	&& \multicolumn{2}{c}{394.566}	&& \multicolumn{2}{c}{398.329}	&&	\multicolumn{2}{c}{394.815}&&	\multicolumn{2}{c}{394.581}&&	\multicolumn{2}{c}{394.553}	\\ \hline
\end{tabular}
}

\end{table}

\begin{table}[p]
 \caption{Optimal BM relativities and stationary distribution of $L$ for various threshold $\varphi$ and $\sigma^{[2]}$ under the two BMSs: the -1/+1 and -1/+1/+2 systems} \label{tab.ex4}
 \centering
 \resizebox{\linewidth}{!}{
 \begin{tabular}{ l c c c c c c c c c c c c c c c c c c c }
 \\
 \multicolumn{12}{l}{(a) -1/+1 system and -1/+1/+2 system for $\left(\sigma^{[2]}\right)^2=0.01$}\\

 \hline
&& \multicolumn{2}{l}{-1/+1 system} && \multicolumn{6}{l}{-1/+1/+2 system}\\\cline{3-4} \cline{6-16}
 & $\varphi$& \multicolumn{2}{c}{-} && \multicolumn{2}{c}{9000} && \multicolumn{2}{c}{16800} && \multicolumn{2}{c}{38000}&& \multicolumn{2}{c}{60400} \\\cline{3-4} \cline{6-7}\cline{9-10}\cline{12-13}\cline{15-16}
Level $\ell$ && $r(\ell)$ & $\P{L=\ell}$ && $r(\ell)$ & $\P{L=\ell}$ && $r(\ell)$ & $\P{L=\ell}$ && $r(\ell)$ & $\P{L=\ell}$ && $r(\ell)$ & $\P{L=\ell}$ \\
 \hline
9	&&	3.070 	&	0.135 	&&	2.852 	&	0.154 	&&	2.976 	&	0.143 	&&	3.060 	&	0.136 	&&	3.069 	&	0.135 	\\
8	&&	1.951 	&	0.055 	&&	1.786 	&	0.066 	&&	1.879 	&	0.059 	&&	1.944 	&	0.055 	&&	1.951 	&	0.055 	\\
7	&&	1.572 	&	0.034 	&&	1.416 	&	0.042 	&&	1.503 	&	0.037 	&&	1.564 	&	0.034 	&&	1.571 	&	0.034 	\\
6	&&	1.344 	&	0.026 	&&	1.193 	&	0.033 	&&	1.276 	&	0.029 	&&	1.337 	&	0.026 	&&	1.343 	&	0.026 	\\
5	&&	1.174 	&	0.024 	&&	1.027 	&	0.030 	&&	1.107 	&	0.027 	&&	1.167 	&	0.024 	&&	1.174 	&	0.024 	\\
4	&&	1.029 	&	0.026 	&&	0.886 	&	0.033 	&&	0.963 	&	0.029 	&&	1.021 	&	0.027 	&&	1.028 	&	0.026 	\\
3	&&	0.888 	&	0.034 	&&	0.761 	&	0.041 	&&	0.827 	&	0.038 	&&	0.881 	&	0.035 	&&	0.888 	&	0.035 	\\
2	&&	0.741 	&	0.055 	&&	0.628 	&	0.062 	&&	0.688 	&	0.058 	&&	0.735 	&	0.056 	&&	0.740 	&	0.055 	\\
1	&&	0.579 	&	0.114 	&&	0.531 	&	0.094 	&&	0.558 	&	0.105 	&&	0.576 	&	0.113 	&&	0.578 	&	0.114 	\\
0	&&	0.396 	&	0.496 	&&	0.367 	&	0.445 	&&	0.384 	&	0.474 	&&	0.395 	&	0.494 	&&	0.396 	&	0.496 	\\
\hline\hline
HMSE	&& \multicolumn{2}{c}{9.480}	&& \multicolumn{2}{c}{10.158}	&&	\multicolumn{2}{c}{9.760}&&	\multicolumn{2}{c}{9.507}&&	\multicolumn{2}{c}{9.482}	\\ \hline
\end{tabular}
}
 \bigskip

 \resizebox{\linewidth}{!}{
 \begin{tabular}{ l c c c c c c c c c c c c c c c c c c c }
 \multicolumn{12}{l}{(b) -1/+1 system and -1/+1/+2 system for $\left(\sigma^{[2]}\right)^2=0.29$}\\

 \hline
&& \multicolumn{2}{l}{-1/+1 system} && \multicolumn{6}{l}{-1/+1/+2 system}\\\cline{3-4} \cline{6-16}
 & $\varphi$& \multicolumn{2}{c}{-} && \multicolumn{2}{c}{8200} && \multicolumn{2}{c}{16800} && \multicolumn{2}{c}{48100}&& \multicolumn{2}{c}{94300} \\\cline{3-4} \cline{6-7}\cline{9-10}\cline{12-13}\cline{15-16}
Level $\ell$ && $r(\ell)$ & $\P{L=\ell}$ && $r(\ell)$ & $\P{L=\ell}$ && $r(\ell)$ & $\P{L=\ell}$ && $r(\ell)$ & $\P{L=\ell}$ && $r(\ell)$ & $\P{L=\ell}$ \\
 \hline
9	&&	3.070 	&	0.135 	&&	2.893 	&	0.154 	&&	3.009 	&	0.143 	&&	3.069 	&	0.136 	&&	3.071 	&	0.135 	\\
8	&&	1.951 	&	0.055 	&&	1.827 	&	0.066 	&&	1.915 	&	0.059 	&&	1.954 	&	0.055 	&&	1.952 	&	0.055 	\\
7	&&	1.572 	&	0.034 	&&	1.451 	&	0.042 	&&	1.535 	&	0.037 	&&	1.574 	&	0.034 	&&	1.573 	&	0.034 	\\
6	&&	1.344 	&	0.026 	&&	1.220 	&	0.033 	&&	1.303 	&	0.029 	&&	1.345 	&	0.026 	&&	1.345 	&	0.026 	\\
5	&&	1.174 	&	0.024 	&&	1.046 	&	0.030 	&&	1.127 	&	0.027 	&&	1.174 	&	0.024 	&&	1.175 	&	0.024 	\\
4	&&	1.029 	&	0.026 	&&	0.898 	&	0.033 	&&	0.976 	&	0.029 	&&	1.026 	&	0.027 	&&	1.029 	&	0.026 	\\
3	&&	0.888 	&	0.034 	&&	0.761 	&	0.040 	&&	0.831 	&	0.037 	&&	0.883 	&	0.035 	&&	0.888 	&	0.035 	\\
2	&&	0.741 	&	0.055 	&&	0.625 	&	0.062 	&&	0.684 	&	0.058 	&&	0.733 	&	0.056 	&&	0.740 	&	0.055 	\\
1	&&	0.579 	&	0.114 	&&	0.499 	&	0.094 	&&	0.531 	&	0.105 	&&	0.569 	&	0.113 	&&	0.577 	&	0.114 	\\
0	&&	0.396 	&	0.496 	&&	0.352 	&	0.447 	&&	0.371 	&	0.475 	&&	0.392 	&	0.494 	&&	0.396 	&	0.496 	\\
\hline\hline
HMSE	&& \multicolumn{2}{c}{19.120}	&& \multicolumn{2}{c}{19.345}	&&	\multicolumn{2}{c}{19.025}&&	\multicolumn{2}{c}{19.049}&&	\multicolumn{2}{c}{19.105}	\\ \hline
\end{tabular}
}

 \bigskip

 \resizebox{\linewidth}{!}{
 \begin{tabular}{ l c c c c c c c c c c c c c c c c c c c }
 \multicolumn{12}{l}{(b) -1/+1 system and -1/+1/+2 system for $\left(\sigma^{[2]}\right)^2=1$}\\

 \hline
&& \multicolumn{2}{l}{-1/+1 system} && \multicolumn{6}{l}{-1/+1/+2 system}\\\cline{3-4} \cline{6-16}
 & $\varphi$& \multicolumn{2}{c}{-} && \multicolumn{2}{c}{6400} && \multicolumn{2}{c}{16100} && \multicolumn{2}{c}{68100}&& \multicolumn{2}{c}{182300} \\\cline{3-4} \cline{6-7}\cline{9-10}\cline{12-13}\cline{15-16}
Level $\ell$ && $r(\ell)$ & $\P{L=\ell}$ && $r(\ell)$ & $\P{L=\ell}$ && $r(\ell)$ & $\P{L=\ell}$ && $r(\ell)$ & $\P{L=\ell}$ && $r(\ell)$ & $\P{L=\ell}$ \\
 \hline
9	&&	3.070 	&	0.135 	&&	2.968 	&	0.154 	&&	3.078 	&	0.143 	&&	3.092 	&	0.136 	&&	3.076 	&	0.135 	\\
8	&&	1.951 	&	0.055 	&&	1.901 	&	0.066 	&&	1.987 	&	0.059 	&&	1.980 	&	0.055 	&&	1.958 	&	0.055 	\\
7	&&	1.572 	&	0.034 	&&	1.512 	&	0.041 	&&	1.598 	&	0.037 	&&	1.598 	&	0.034 	&&	1.578 	&	0.034 	\\
6	&&	1.344 	&	0.026 	&&	1.266 	&	0.032 	&&	1.352 	&	0.029 	&&	1.365 	&	0.026 	&&	1.350 	&	0.026 	\\
5	&&	1.174 	&	0.024 	&&	1.073 	&	0.030 	&&	1.161 	&	0.027 	&&	1.189 	&	0.024 	&&	1.179 	&	0.024 	\\
4	&&	1.029 	&	0.026 	&&	0.914 	&	0.032 	&&	0.995 	&	0.029 	&&	1.035 	&	0.027 	&&	1.031 	&	0.026 	\\
3	&&	0.888 	&	0.034 	&&	0.750 	&	0.040 	&&	0.827 	&	0.037 	&&	0.885 	&	0.035 	&&	0.889 	&	0.034 	\\
2	&&	0.741 	&	0.055 	&&	0.616 	&	0.062 	&&	0.675 	&	0.058 	&&	0.729 	&	0.056 	&&	0.739 	&	0.055 	\\
1	&&	0.579 	&	0.114 	&&	0.439 	&	0.095 	&&	0.480 	&	0.106 	&&	0.551 	&	0.113 	&&	0.573 	&	0.114 	\\
0	&&	0.396 	&	0.496 	&&	0.322 	&	0.448 	&&	0.346 	&	0.476 	&&	0.383 	&	0.494 	&&	0.394 	&	0.496 	\\
\hline\hline
HMSE	&& \multicolumn{2}{c}{59.706}	&& \multicolumn{2}{c}{59.053}	&&	\multicolumn{2}{c}{58.814}&&	\multicolumn{2}{c}{59.365}&&	\multicolumn{2}{c}{59.631}	\\ \hline
\end{tabular}
}

\end{table}


\begin{table}[p]
\centering
\caption{(Data analysis) Observable policy characteristics used as covariates} \label{tab.x}\vspace{.05in}
\begin{tabular}{l|l r r r r r r r }
\hline
Categorical & \multirow{2}{*}{Description} && \multicolumn{3}{c}{\multirow{2}{*}{Proportions}} \\
variables & & & & \\
\hline
Entity type & Type of local government entity \\
		& \quad\quad\quad\quad\quad\quad Miscellaneous 	&& \multicolumn{3}{c}{5.03$\%$} \\
		& \quad\quad\quad\quad\quad\quad City			&& \multicolumn{3}{c}{9.66$\%$} \\
		& \quad\quad\quad\quad\quad\quad County			&& \multicolumn{3}{c}{11.47$\%$} \\
		& \quad\quad\quad\quad\quad\quad School			&& \multicolumn{3}{c}{36.42$\%$} \\
		& \quad\quad\quad\quad\quad\quad Town			&& \multicolumn{3}{c}{16.90$\%$} \\
		& \quad\quad\quad\quad\quad\quad Village 			&& \multicolumn{3}{c}{20.52$\%$} \\
\hline
Coverage & Collision coverage amount for old and new vehicles\\
		& \quad\quad\quad\quad\quad\quad Coverage $\in (0,0.14] = 1 $ && \multicolumn{3}{c}{33.40$\%$} \\
		& \quad\quad\quad\quad\quad\quad Coverage $\in (0.14,0.74] = 2 $	&& \multicolumn{3}{c}{33.20$\%$} \\
		& \quad\quad\quad\quad\quad\quad Coverage $\in (0.74,\infty) = 3$	&& \multicolumn{3}{c}{33.40$\%$} \\
\hline
\end{tabular}
\end{table}


\bigskip

\begin{table}[p]
\caption{(Data analysis) Estimation results under the frequency--severity model }\vspace{.05in}\label{est.model2}
\centering
\begin{tabular}{ l r r r r l r r r r l c c }
 \hline
&& &\multicolumn{2}{c}{95$\%$ CI}& \\
parameter& Est & Std.dev & lower & upper& \\
 \hline
 \multicolumn{3}{l}{ \textbf{Frequency part}} \\
\quad Intercept &	-2.767& 	0.318& -3.417& -2.153&*&\\
\quad City & 	 0.597& 	0.337& -0.051& 	1.272& &\\
\quad County & 1.907& 	0.335& 	1.271& 	2.587&*&	\\
\quad School &	 0.411& 	0.304& -0.181& 	1.014& &\\
\quad Town &	-1.351& 	0.384& -2.103& -0.584&*&\\
\quad Village &	-0.012& 	0.323& -0.626& 	0.654& &\\
\quad Coverage2 &	 1.247& 	0.212& 	0.829& 	1.667&*& \\
\quad Coverage3 & 2.139& 	0.230& 	1.713& 	2.615&*& 	\\
\hline
 \multicolumn{3}{l}{ \textbf{Severity part}} \\

\quad Intercept & 	 8.829& 	0.375& 	8.103& 	9.588&*& 	 \\
\quad City 	 & 		-0.036& 	0.353& -0.737& 	0.637& & \\
\quad County 	 & 		 0.341& 	0.338& -0.336& 	0.980& & \\
\quad School 	 & 		-0.173& 	0.328& -0.805& 	0.484& & \\
\quad Town 	 &		 0.497& 	0.440& -0.356& 	1.349& & \\
\quad Village &		 0.316& 	0.346& -0.357& 	0.994& & \\
\quad Coverage2 &		 0.180& 	0.244& -0.308& 	0.646& & 	\\
\quad Coverage3 & 		-0.027& 	0.261& -0.533& 	0.493& & \\
\quad $1/\psi^{[2]}$& 0.670& 	0.041& 	0.592& 	0.752&*& 	\\
\hline
\multicolumn{3}{l}{ \textbf{Copula part}} \\

\quad $\left(\sigma^{[1]}\right)^2$& 	 0.992& 	0.142& 	0.746& 	1.292&*& \\
\quad $\left(\sigma^{[2]}\right)^2$& 	 0.293& 	0.067& 	0.176& 	0.433&*&\\
\quad $\rho$& 	-0.447& 	0.130& -0.690& -0.190&*\\
 \hline
\end{tabular}
\end{table}


\begin{table}[p]
 \caption{ (Data analysis) Optimal BM relativities and stationary distribution of $L$ for various threshold $\varphi$ and under the different BMSs} \label{tab.dat}
 \centering
 \resizebox{\linewidth}{!}{

 \begin{tabular}{ l c c c c c c c c c c c c c c c c c c c }
 \\
 \multicolumn{12}{l}{(a) -1/+1/+1 system and -1/+1/+2 system }\\
 \hline
&& \multicolumn{2}{l}{-1/+1/+1 system} && \multicolumn{6}{l}{-1/+1/+2 system}\\\cline{3-4} \cline{6-16}
 & $\varphi$& \multicolumn{2}{c}{-} && \multicolumn{2}{c}{10100} && \multicolumn{2}{c}{21400} && \multicolumn{2}{c}{64500}&& \multicolumn{2}{c}{132100} \\\cline{3-4} \cline{6-7}\cline{9-10}\cline{12-13}\cline{15-16}
Level $\ell$ && $r(\ell)$ & $\P{L=\ell}$ && $r(\ell)$ & $\P{L=\ell}$ && $r(\ell)$ & $\P{L=\ell}$ && $r(\ell)$ & $\P{L=\ell}$ && $r(\ell)$ & $\P{L=\ell}$ \\
 \hline
9	&&	0.966 	&	0.146 	&&	0.952 	&	0.153 	&&	0.959 	&	0.148 	&&	0.965 	&	0.146 	&&	0.966 	&	0.146 	\\
8	&&	0.454 	&	0.034 	&&	0.441 	&	0.038 	&&	0.448 	&	0.035 	&&	0.453 	&	0.034 	&&	0.454 	&	0.034 	\\
7	&&	0.369 	&	0.018 	&&	0.354 	&	0.021 	&&	0.361 	&	0.019 	&&	0.368 	&	0.018 	&&	0.369 	&	0.018 	\\
6	&&	0.327 	&	0.013 	&&	0.312 	&	0.016 	&&	0.318 	&	0.014 	&&	0.326 	&	0.013 	&&	0.327 	&	0.013 	\\
5	&&	0.302 	&	0.012 	&&	0.286 	&	0.014 	&&	0.292 	&	0.013 	&&	0.300 	&	0.012 	&&	0.302 	&	0.012 	\\
4	&&	0.284 	&	0.013 	&&	0.268 	&	0.016 	&&	0.273 	&	0.014 	&&	0.282 	&	0.013 	&&	0.284 	&	0.013 	\\
3	&&	0.271 	&	0.017 	&&	0.256 	&	0.021 	&&	0.260 	&	0.018 	&&	0.268 	&	0.017 	&&	0.270 	&	0.017 	\\
2	&&	0.259 	&	0.029 	&&	0.248 	&	0.037 	&&	0.250 	&	0.032 	&&	0.256 	&	0.029 	&&	0.259 	&	0.029 	\\
1	&&	0.250 	&	0.075 	&&	0.243 	&	0.066 	&&	0.243 	&	0.071 	&&	0.247 	&	0.074 	&&	0.250 	&	0.075 	\\
0	&&	0.243 	&	0.645 	&&	0.239 	&	0.619 	&&	0.239 	&	0.635 	&&	0.241 	&	0.645 	&&	0.242 	&	0.645 	\\
\hline\hline
HMSE	&& \multicolumn{2}{c}{91.242}	&& \multicolumn{2}{c}{92.128}	&&	\multicolumn{2}{c}{91.634}&&	\multicolumn{2}{c}{91.269}&&	\multicolumn{2}{c}{91.241}	\\ \hline
\end{tabular}
}
 \bigskip

 \resizebox{\linewidth}{!}{
 \begin{tabular}{ l c c c c c c c c c c c c c c c c c c c }
 \multicolumn{12}{l}{(b) -1/+2/+2 system and -1/+2/+3 system }\\
 \hline
&& \multicolumn{2}{l}{-1/+2/+2 system} && \multicolumn{6}{l}{-1/+2/+3 system}\\\cline{3-4} \cline{6-16}
 & $\varphi$& \multicolumn{2}{c}{-} && \multicolumn{2}{c}{10100} && \multicolumn{2}{c}{21400} && \multicolumn{2}{c}{64500}&& \multicolumn{2}{c}{132100} \\\cline{3-4} \cline{6-7}\cline{9-10}\cline{12-13}\cline{15-16}
Level $\ell$ && $r(\ell)$ & $\P{L=\ell}$ && $r(\ell)$ & $\P{L=\ell}$ && $r(\ell)$ & $\P{L=\ell}$ && $r(\ell)$ & $\P{L=\ell}$ && $r(\ell)$ & $\P{L=\ell}$ \\
 \hline
9	&&	0.932 	&	0.178 	&&	0.927 	&	0.183 	&&	0.929 	&	0.180 	&&	0.931 	&	0.178 	&&	0.932 	&	0.178 	\\
8	&&	0.428 	&	0.051 	&&	0.424 	&	0.054 	&&	0.426 	&	0.053 	&&	0.428 	&	0.052 	&&	0.428 	&	0.051 	\\
7	&&	0.343 	&	0.030 	&&	0.339 	&	0.033 	&&	0.341 	&	0.031 	&&	0.343 	&	0.030 	&&	0.343 	&	0.030 	\\
6	&&	0.302 	&	0.023 	&&	0.298 	&	0.025 	&&	0.300 	&	0.024 	&&	0.301 	&	0.023 	&&	0.302 	&	0.023 	\\
5	&&	0.278 	&	0.020 	&&	0.274 	&	0.023 	&&	0.276 	&	0.021 	&&	0.278 	&	0.020 	&&	0.278 	&	0.020 	\\
4	&&	0.262 	&	0.024 	&&	0.260 	&	0.025 	&&	0.260 	&	0.025 	&&	0.261 	&	0.024 	&&	0.262 	&	0.024 	\\
3	&&	0.253 	&	0.024 	&&	0.251 	&	0.031 	&&	0.251 	&	0.027 	&&	0.253 	&	0.025 	&&	0.253 	&	0.024 	\\
2	&&	0.245 	&	0.055 	&&	0.244 	&	0.051 	&&	0.244 	&	0.054 	&&	0.245 	&	0.055 	&&	0.245 	&	0.055 	\\
1	&&	0.243 	&	0.046 	&&	0.242 	&	0.043 	&&	0.242 	&	0.045 	&&	0.242 	&	0.046 	&&	0.243 	&	0.046 	\\
0	&&	0.239 	&	0.547 	&&	0.238 	&	0.531 	&&	0.238 	&	0.541 	&&	0.238 	&	0.546 	&&	0.239 	&	0.547 	\\
\hline\hline
HMSE	&& \multicolumn{2}{c}{93.483}	&&\multicolumn{2}{c}{93.786}&& \multicolumn{2}{c}{93.625}	&&	\multicolumn{2}{c}{93.498}&&	\multicolumn{2}{c}{93.484}	\\ \hline
\end{tabular}
}

\end{table}


\end{document}